\documentclass[11pt,letterpaper]{article}

\usepackage{typearea}
\typearea{14}

\usepackage{comment,algorithm,algorithmic,multicol}


\usepackage{color,colortbl}
\usepackage{float}
\usepackage{amsthm}
\usepackage{amsmath}
\usepackage{amssymb}
\usepackage{graphicx}
\usepackage{xspace}
\usepackage{nicefrac}
\usepackage[colorinlistoftodos,prependcaption,textsize=tiny]{todonotes}

\usepackage{thmtools}
\declaretheorem[name=Lemma]{lemma}

\definecolor{Darkblue}{rgb}{0,0,0.4}
\definecolor{Brown}{cmyk}{0,0.61,1.,0.60}
\definecolor{Purple}{cmyk}{0.45,0.86,0,0}

\usepackage[colorlinks,linkcolor=Darkblue,filecolor=blue,citecolor=blue,urlcolor=Darkblue,pagebackref]{hyperref}
\usepackage[nameinlink]{cleveref}

\newtheorem{theorem}{Theorem}
\newtheorem{corollary}{Corollary}
\newtheorem{remark}{Remark}
\newtheorem{claim}{Claim}
\newtheorem{fact}{Fact}

\newcommand{\namedref}[2]{\hyperref[#2]{#1~\ref*{#2}}}
\newcommand{\subsectionref}[1]{\namedref{Subsection}{#1}}

\newcommand{\lineref}[1]{\namedref{Line}{#1}}

\newcommand{\R}{\mathbb{R}}

\newcommand{\Exp}{\mathsf{Exp}}
\newcommand{\Geo}{\mathsf{Geo}}

\newcommand{\eps}{\epsilon}

\newcommand{\la}{~\leftarrow~}

\newcommand{\cint}{c_{\text{\tiny int}}}

\newcommand{\EfBig}{\mathcal{E}^{\text{\tiny fBig}}}

\newcommand{\EB}{\mathcal{E}^{\text{\tiny B}}}

\newcommand{\NV}{\texttt{Noisy-Voronoi}\xspace}

\usepackage{pdfsync}

\title{Steiner Point Removal with distortion $O(\log k)$,\\ using the \NV algorithm\thanks{A preliminary version was published at SODA'18 \cite{Fil18}.}}
\author{Arnold Filtser\\Ben Gurion University of the Negev\\
	Email: \texttt{arnoldf@cs.bgu.ac.il}}
\date{\today}

\begin{document}
	\maketitle
	\thispagestyle{empty}
	\nonumber
	\begin{abstract}
In the Steiner Point Removal (SPR) problem, we are given a weighted graph $G=(V,E)$ and a set of terminals $K\subset V$ of size $k$.
The objective is to find a minor $M$ of $G$ with only the terminals as its vertex set, such that distances between the terminals will be preserved up to a small multiplicative distortion.
Kamma, Krauthgamer and Nguyen [SICOMP2015] devised a ball-growing algorithm with exponential distributions to show that the distortion is at most $O(\log^5 k)$.
Cheung [SODA2018] improved the analysis of the same algorithm, bounding the distortion by $O(\log^2 k)$.
We devise a novel and simpler algorithm (called the Noisy Voronoi algorithm) which incurs distortion $O(\log k)$. This algorithm can be implemented in almost linear time ($O(|E|\log |V|)$). 
	\end{abstract}

\newpage
\pagenumbering{arabic}
\section{Introduction}
In graph compression problems the input is usually a massive graph. The objective is to compress the graph into a smaller graph, while preserving certain  properties of the original graph, such as distances or cut values.
Compression allows us to obtain faster algorithms, while reducing the storage space. In the era of massive data, the benefits are obvious.
Examples of such structures are graph spanners \cite{PS89}, distance oracles \cite{TZ05}, cut sparsifiers \cite{BK96}, spectral sparsifiers \cite{BSS12}, vertex sparsifiers \cite{Moitra09} and more.

In this paper we study the \emph{Steiner point removal} (SPR) problem. Here we are given an undirected graph $G=(V,E)$ with positive weight function $w:E\rightarrow\mathbb{R}_+$, and a subset of terminals $K\subseteq V$ of size $k$ (the non-terminal vertices are called Steiner vertices).
The goal is to construct a new graph $M=(K,E')$ with positive weight function $w'$, with the terminals as its vertex set, such that: (1) $M$ is a graph minor of $G$, and (2) the distance between every pair of terminals $t,t'$ is distorted by at most a multiplicative factor of $\alpha$, formally
$$\forall t,t'\in K,~~d_G(t,t')\le d_{M}(t,t')\le \alpha \cdot d_G(t,t')~.$$
Property (1) expresses preservation of the topological structure of the original graph. For example if $G$ was planar, so will $M$ be. Whereas property (2) expresses preservation of the geometric structure of the original graph, that is, distances between terminals.
The question is: what is the minimal $\alpha$ (which may depend on $k$) such that every graph with a terminal set of size $k$ will admit a solution to the SPR problem with distortion $\alpha$.

The first one to study a problem of this flavor was Gupta \cite{G01}, who showed that given a weighted tree $T$ with a subset of terminals $K$, there is a tree $T'$ with $K$ as its vertex set, that preserves all the distances between terminals up to a multiplicative factor of $8$.
Chan, Xia, Konjevod, and Richa \cite{CXKR06}, observed that the tree $T'$ of Gupta is in fact a minor of the original tree $T$. They showed that $8$ is the best possible distortion, and formulated the problem for general graphs.
This lower bound of $8$ is achieved on the complete unweighted binary tree, and is the best known lower bound for the general SPR problem. 

Basu and Gupta \cite{BG08} showed that on outerplanar graphs, the SPR problem can be solved with distortion $O(1)$.

Kamma, Krauthgamer and Nguyen were the first to bound the distortion for general graphs.
They suggested the \texttt{Ball-growing} algorithm.   
Their first analysis provide  $O(\log^6 k)$ distortion (conference version \cite{KKN14}), which they later improved to $O(\log^5 k)$  (journal version \cite{KKN15}).
Recently, Cheung \cite{Che18} improved the analysis of the  \texttt{Ball-growing} algorithm further, providing an $O(\log^{2}k)$ upper bound on the distortion.

The \texttt{Ball-growing} algorithm constructs a terminal partition, that is a partition where each cluster is connected and contains a single terminal. The minor is then constructed by contracting all the internal edges in all clusters. The weight of the minor edge $\{t,t'\}$ (if exist) defined simply to $d_G(t,t')$.
The clusters are generated iteratively. In each round, by turn, each terminal $t_j$ increases the radius $R_j$ of its ball-cluster $V_j$ in an attempt to add more vertices to its ball cluster $V_j$. Once a vertex joins some cluster, it will remain there. 
In round $\ell$, the radii are (independently) distributed according to an exponential distribution, where the mean of the distribution grows in each round. A description of the \texttt{Ball-growing}  algorithm could be found in \Cref{sec:BallGrowing}.

The main contribution of this paper is a new upper bound of $O(\log k)$ for the Steiner Point Removal problem.
In a preliminary conference version \cite{Fil18}, the author improved the analysis of the \texttt{Ball-growing} algorithm, providing an $O(\log k)$ upper bound.
In this paper we devise a novel algorithm called the \texttt{Noisy-Voronoi} algorithm. We bound the distortion incurred by the minor produced using the \texttt{Noisy-Voronoi} by $O(\log k)$ as well.
Nevertheless, the \texttt{Noisy-Voronoi} algorithm is arguably simpler and more intuitive compared to the \texttt{Ball-growing} algorithm. 
Both algorithms grow clusters around the terminals, the main difference is that the \texttt{Ball-growing} algorithm has many iterations, growing slowly from all terminals (almost in parallel), while the \texttt{Noisy-Voronoi} algorithm has one round only (each terminal construct a cluster by turn and done.
The analysis in \cite{Fil18} was built upon \cite{Che18}. In both papers, a considerable effort was made to lower and upper bound the number of the round in which each non-terminal is clustered.
The analysis in this paper is quite similar to \cite{Fil18}, while all the round-base analysis simply becomes unnecessary. 

Furthermore, we devise an efficient implementation of the \texttt{Noisy-Voronoi} algorithm in almost linear time $O\left(m+\min\{m,nk\}\cdot\log n\right)$ ($m$ (resp. $n$) here is the number of edges (resp. vertices) in $G$).
While the \texttt{Ball-growing} algorithm can be implemented in polynomial time, it is not clear how to do so efficiently. 
 
We show that the analysis of the \texttt{Noisy-Voronoi} algorithm is asymptotically tight. That is, there are graphs for which the \texttt{Noisy-Voronoi} produces a minor which incur distortion $\Omega(\log k)$. We prove a similar lower bound also for the  \texttt{Ball-growing} algorithm. However, there we are only able to prove a $\Omega(\sqrt{\log k})$ lower bound on the performance of the algorithm.

\subsection{Related Work}
Englert et. al. \cite{EGKRTT14} showed that every graph $G$, admits a distribution $\mathcal{D}$ over terminal minors with expected distortion $O(\log k)$. Formally, for all $t_{i},t_{j}\in K$, it holds that $1\le\frac{\mathbb{E}_{M\sim\mathcal{D}}\left[d_{M}(t_{i},t_{j})\right]}{d_{G}(t_{i},t_{j})}\le O\left(\log k\right)$. 
Thus, \Cref{thm:mainSPR} can be seen as improvement upon \cite{EGKRTT14}, where we replace distribution with a single minor.
Englert et. al. showed better results for  $\beta$-decomposable graphs, in particular, they showed that graphs excluding a fixed minor admit a distribution with $O(1)$ expected distortion.

Krauthgamer, Nguyen and Zondiner \cite{KNZ14} showed that if we allow the minor $M$ to contain at most ${k\choose 2}^2$ Steiner vertices (in addition to the terminals), then distortion $1$ can be achieved. They further showed that for graphs with constant treewidth, $O(k^2)$ Steiner points will suffice for distortion $1$.
Cheung, Gramoz and Henzinger \cite{CGH16} showed that allowing $O(k^{2+\frac2t})$ Steiner vertices, one can achieve distortion $2t-1$ (in particular distortion $O(\log k)$ with $O(k^2)$ Steiners). For planar graphs, Cheung et. al. achieved $1+\eps$ distortion with $\tilde{O}((\frac k\epsilon)^2)$ Steiner points.

There is a long line of work focusing on preserving the cut/flow structure among the terminals by a graph minor. See 
\cite{Moitra09,LM10,CLLM10,MM10,EGKRTT14,Chuzhoy12,KR13,AGK14,GHP17,KR17}.

There were works studying metric embeddings and metric data structures concerning with preserving distances among terminals, or from terminals to other vertices, out of the context of minors. See \cite{CE05,RTZ05,GNR10,KV13,EFN15S,EFN17,BFN16}.

Finally, there are clustering algorithms which are similar in nature to the \texttt{Noisy-Voronoi} and \texttt{Ball-growing} algorithms \cite{LS91,Bar96,FRT04,CKR01,FHRT03,MPVX15}.

\subsection{Technical Ideas}
The basic approach in this paper, as well as in all previous papers on SPR in general graphs, is to use terminal partitions in order to construct a minor for the SPR problem.
Specifically, we partition the vertices into $k$ connected clusters, with a single terminal in each cluster. Such a partition induces a minor by contracting all the internal edges in each cluster. See the preliminaries for more details. 
Considering such a framework, the most natural idea will be to partition the vertices into the Voronoi cells. i.e., the cluster $V_j$ of the terminal $t_j$ will contain all the vertices $v$ for which $t_j$ is the closest terminal.
However, this approach miserably fails and can incur distortion as large as $k-1$. See \Cref{fig:VoronoiFail} for illustration. 
\begin{figure}[]
	\centering{\includegraphics[scale=0.5]{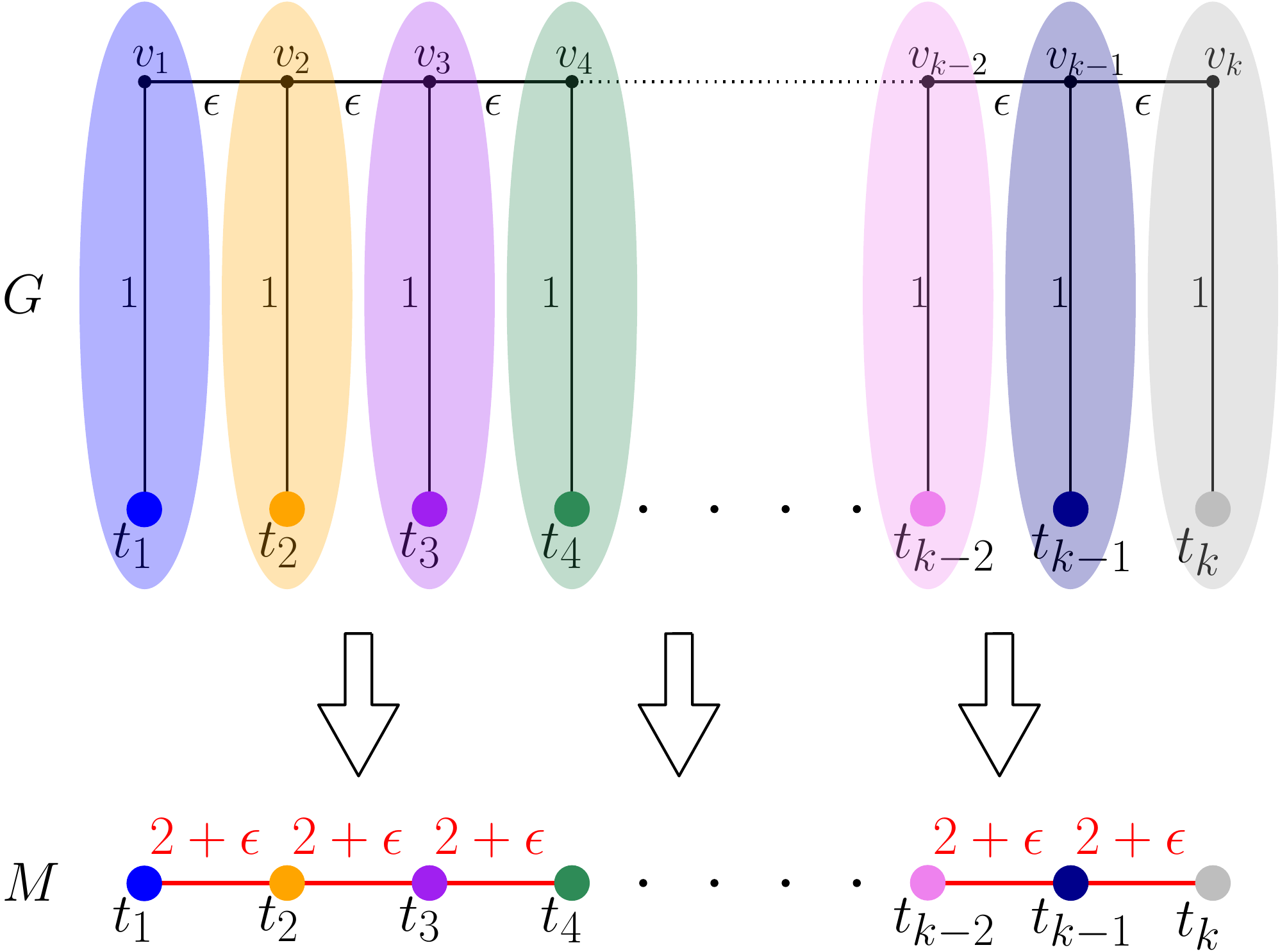}} 
	\caption{\label{fig:VoronoiFail}\small \it 
		The graph $G$ consist of a $k$-path of Steiner vertices $v_1,\dots,v_k$ with edges of weight $\eps$. To each Steiner vertex $v_j$ we add a terminal using a unit weight edge.
		The Voronoi cell of the terminal $t_j$ is $\{t_j,v_j\}$. The minor $M$ induced by this terminal partition is a path $t_1,\dots,t_k$ where the weight of each edge equals $2+\eps$. The original distance in $G$ between $t_1$ to $t_k$ is $2+(k-1)\cdot\eps$, while the distance in the minor $M$ equals $(k-1)\cdot(2+\eps)$. In particular, when $\eps$ tends to $0$, the distortion tends to $k-1$.
	}
\end{figure}

Our idea is to introduce some noise in order to avoid the sharp boundaries between the clusters. Specifically, we order the terminals in an arbitrary order. For each terminal $t_j$ we sample a parameter $R_j\ge 1$ that we will call its \emph{magnitude}. Then, by turn, each terminal will construct a cluster $V_j$ which will be essentially a magnified (by $R_j$) Voronoi cell (in the remaining graph).
However, in order to maintain connectivity, the magnified Voronoi cell is constructed in a ``Dijkstra manner'' 
as follows. For every vertex $v$, denote by $D(v)$ the distance from $v$ to its closest terminal. Initially $V_j=\{t_j\}$. In each step, every unclustered neighboring vertex $v$ of $V_j$ is examined. If $d_G(v,t_j)\le R_j\cdot D(v)$, then $v$ joins the cluster $V_j$. The process terminates when no new potential vertices remain. Then we move on to the next terminal and repeat the same process on the remaining graph. Eventually, all of $G$ is partitioned into clusters.

To sample $R_j$, we first sample $g_j$ according to geometric distribution with parameter $p=\frac15$. Then, $R_j$ set to be $(1+\delta)^{g_j}$ where $\delta=\Theta(\frac{1}{\ln k})$. In particular, all the $R_j$'s are bounded by some universal  constant w.h.p.

Next, we provide some intuition for the distortion analysis. Consider a pair of terminals $t,t'$, and let $P_{t,t'}$ be the shortest path between them in the original graph $G$. When the algorithm terminates, all the vertices in $P_{t,t'}$ are clustered by different terminals. See \Cref{fig:distortion} for illustration. Let $\mathcal{D}_{\ell_1},\dots,\mathcal{D}_{\ell_k}$ be the partition of the vertices in $P_{t,t'}$ induced by the partition of all vertices created by the algorithm. i.e., $\mathcal{D}_{\ell_i}=P_{t,t'}\cap V_{\ell_i}$
For simplicity at this stage, we will assume that every $\mathcal{D}_{\ell_j}$ is continuous. 
In the induced minor graph, there is an edge between any two consecutive terminals $t_{\ell_j}$ and $t_{\ell_{j+1}}$.
Therefore the distance between $t$ to $t'$ in the minor graph can be bounded by $\sum_{j}d_{G}(t_{\ell_{j}},t_{\ell_{j+1}})$.
Let $v^{\ell_j}$ be the ``first'' vertex on $P_{t,t'}$ to be covered by $t_{\ell_j}$. ``First'' here is in the following sense: we think on the sampling of $R_j$ in a consecutive manner. For a vertex $v$, let $r_v$ denote the minimal value of $R_j$ such that $v\in V_j$. Then $v^j$ is defined to be the vertex with the minimal value $r_v$. 
Using the triangle inequality, $d_{G}(t_{\ell_{j}},t_{\ell_{j+1}})\le d_{G}(t_{\ell_{j}},v^{\ell_{j}})+d_{G}(v^{\ell_{j}},v^{\ell_{j+1}})+d_{G}(v^{\ell_{j+1}},t_{\ell_{j+1}})$. Therefore $d_{M}(t,t')\le\sum_{i=1}^{k'-1}d_{G}(v^{\ell_{i}},v^{\ell_{i+1}})+2\sum_{i=1}^{k'}d_{G}(t_{\ell_{i}},v^{\ell_{i}})\le d_G(t,t')+2\sum_{i=1}^{k'}d_{G}(t_{\ell_{i}},v^{\ell_{i}})$ (see  \Cref{fig:distortion} for an illustration).

In order to bound the distortion, we need to bound the sum of ``deviations'' $\sum_{i=1}^{k'}d_{G}(t_{\ell_{i}},v^{\ell_{i}})$ from the shortest path. However, these deviations are heavily dependent.
Instead of analyzing the deviations directly, we will follow an approach first suggested by \cite{Che18}. We partition the shortest path $P_{t,t'}$ from $t$ to $t'$ into a set of intervals $\mathcal{Q}$, the idea will be to count for each interval $Q$ how many deviation start from this interval (denoted $X(Q)$). Specifically, for each deviation, we will charge the interval in which this deviation was initiated. Afterwards, we will be able to replace the sum of deviations above by a linear combination of the interval charges.

The partition of the shortest path $P_{t,t'}$ into intervals is done such that the length of each interval $Q\in\mathcal{Q}$ will be a $\log k$ fraction of the distance from the interval to its closest terminal. Such interval lengths will ensure the following crucial property: given that  some vertex $v\in Q$ joins the cluster $V_j$ (of the terminal $t_j$), with probability at least $1-p$, all of $Q$ joins $V_j$.

Using this property alone, one can show that the expected charge on each interval is bounded by a constant. This already will imply an $O(\log k)$ distortion on each pair in expectation. However, as we are interested in $O(\log k)$ distortion on all pairs with high probability, a more subtle argument is required.
We couple the interval charges into a series of independent random variables that dominate the interval charges. Then, a concentration bound on the independent variables implies an upper bound on the sum of interval charges, which provides $O(\log k)$ distortion with high probability.

\subsection{Paper Organization}
In \Cref{sec:Alg} we describe the \texttt{Noisy-Voronoi} algorithm and prove some of its basic properties. Then, in \Cref{sec:distortion} we analyze the distortion incurred by the \texttt{Noisy-Voronoi} algorithm.
In \Cref{sec:fastNV} we introduce a small modification to the \texttt{Noisy-Voronoi} algorithm. We prove that the distortion analysis is still valid, and explain how the modified algorithm can be efficiently implemented.
In \Cref{sec:LB} we prove that our analysis of the \texttt{Noisy-Voronoi} algorithm is asymptotically tight (and provide some lower bound on the performance of the \texttt{Ball-growing} algorithm).
Finally, in \Cref{sec:Discussion} we provide some concluding remarks, and discuss further directions.

\section{Preliminaries}
\Cref{appendix:key} contains a summary of all the definitions and notations we use. The reader is encouraged to refer to this index  while reading. 

We consider undirected graphs $G=(V,E)$ with positive edge weights
$w: E \to \R_{\geq 0}$. Let $d_{G}$ denote the shortest path metric in
$G$. 
For a subset of vertices $A\subseteq V$, let $G[A]$ denote the \emph{induced graph} on $A$. 
Fix $K=\{t_1,\dots,t_k\}\subseteq V$ to be a set of \emph{terminals}. For a vertex $v$, $D(v)=\min_{t\in K}d_{G}(v,t)$ is the distance from $v$ to its closest terminal. 
For clarity, we will assume that all metric distances are unique (that is for $\{v,v'\}\ne\{u,u'\}$, $d_G(v,v')\ne d_G(u,u')$). Moreover, we will assume that for every pair $v,u$ there is a unique shortest path. Otherwise, we can introduce arbitrarily small perturbations. 
\begin{figure}[]
	\centering{\includegraphics[scale=0.85]{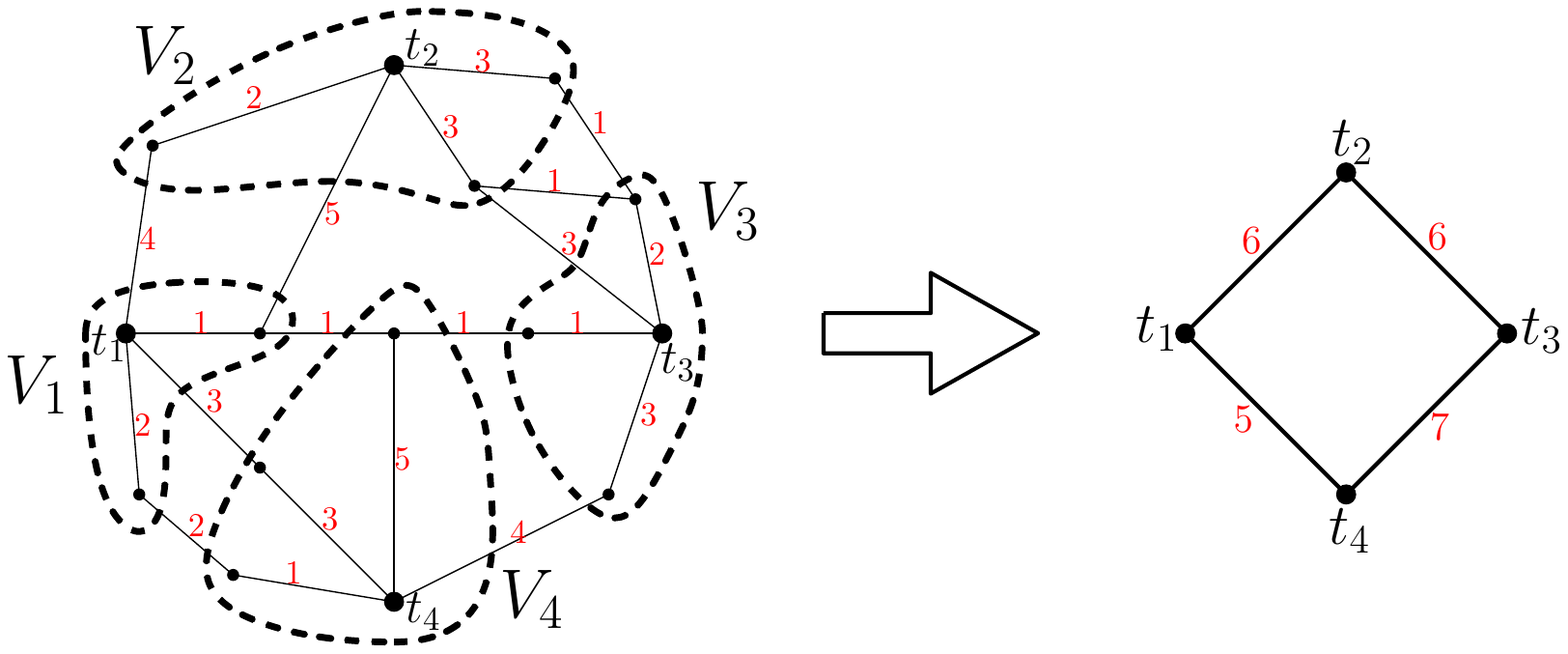}} 
	\caption{\label{fig:contraction}\small \it 
		The left side of the figure contains a weighted graph $G=(V,E)$, with weights specified in red, and four terminals $\{t_1,t_2,t_3,t_4\}$.
		The dashed black curves represent a terminal partition of the vertex set $V$ into the subsets $V_1,V_2,V_3,V_4$.
		The right side of the figure represent the minor $M$ induced by the terminal partition. The distortion is realized between $t_1$ and $t_3$, and is $\frac{d_M(t_1,t_3)}{d_G(t_1,t_3)}=\frac{12}4=3$.
f	}
\end{figure}

A graph $H$ is a \emph{minor} of a graph $G$ if we can obtain $H$ from
$G$ by edge deletions/contractions, and vertex deletions. A partition $\{V_1,\dots,V_k\}$ of $V$ is called a \emph{terminal partition} (w.r.t $K$) if for every $1\le i\le k$, $t_i\in V_i$, and the induced graph $G[V_i]$ is connected. See \Cref{fig:contraction} for an illustration.
The \emph{induced minor} by terminal partition  $\{V_1,\dots,V_k\}$, is a minor $M$, where each set $V_i$ is contracted into a single vertex called (abusing  notation)  $t_i$. 
Note that there is an edge in $M$ from  $t_i$ to $t_j$ iff there are vertices $v_i\in V_i$ and $v_j\in V_j$ such that $\{v_i,v_j\}\in E$. 
We determine the weight of the edge $\{t_i,t_j\}\in E(M)$ to be $d_G(t_i,t_j)$. Note that by the triangle inequality, for every pair of (not necessarily neighboring) terminals $t_i,t_j$, it holds that $d_M(t_i,t_j)\ge d_G(t_i,t_j)$.
The \emph{distortion} of the induced minor is  $\max_{i,j}\frac{d_M(t_i,t_j)}{d_G(t_i,t_j)}$.

\subsection{Probability}
For a distribution $\mathcal{D}$, $X\sim\mathcal{D}$ denotes that $X$ is a random variable distributed according to $\mathcal{D}$.

$\Geo(p)$ denotes the \emph{Geometric distribution} with parameter $p$. Here we toss a biased coin with probability $p$ for heads, until the first time we get heads. $\Geo(p)$ is the number of coin tosses. Formally, $\Geo(p)$ is supported in $\{1,2,3,\dots\}$, where the probability to get $s$ is $(1-p)^{s-1} \cdot p$.

\emph{Exponential distribution} is the continuous analogue of Geometric distribution.
$\Exp(\lambda)$ denotes the Exponential distribution with mean
$\lambda$ and density function $f(x)=\frac{1}{\lambda}e^{-\frac{x}{\lambda}}$ for $x\ge0$. 
Exponential distribution is closed under scaling, that is, for $X\sim \Exp(\lambda)$, $c\cdot X$ is distributed according to $\Exp(c\lambda)$.
We will use the following concentration bound.
\begin{lemma}\label{lem:ExpConcentration}
	Suppose $X_{1},\dots,X_{n}$'s are independent random
	variables, where each $X_{i}$ is distributed according to $\Exp(\lambda_{i})$.
	Let $X=\sum_{i}X_{i}$ and $\lambda_{M}=\max_i\lambda_{i}$. Set $\mu=\mathbb{E}\left[X\right]=\sum_{i}\lambda_{i}$.
	\[
	\text{For }a\ge2\mu\text{, }~~~~~~  \Pr\left[X\ge a\right]\le\exp\left(-\frac{1}{2\lambda_{M}}\left(a-2\mu\right)\right)~.
	\]
\end{lemma}
In \Cref{appendix:ConcentrationBounds} we prove a more general bound. In particular, \Cref{lem:ExpConcentration} above is a special case of \Cref{lem:TightExpConcentration} (which is obtained 
by choosing parameters $\alpha=\frac a\mu-1$ and $t=\frac{1}{2\lambda_M}$).

\section{Algorithm}\label{sec:Alg}
The terminals are ordered in arbitrary order $t_1,t_2,\dots,t_k$. The \texttt{Noisy-Voronoi} algorithm has $k$ rounds, where in the round $i$, the cluster $V_i$ (containing $t_i$) is constructed in the graph induced by the non-terminal vertices not clustered so far. 

The clusters are created using the \texttt{Create-Cluster} procedure. The algorithm provides a random variable $R_j=(1+\delta)^{g_j}$, where $g_j$ is distributed according to geometric distribution with parameter $p$. 

The \texttt{Create-Cluster} procedure runs in a Dijkstra-like fashion. During the execution, we maintain three sets.
(1) $V_j$: the currently created cluster (initiated to be $\{t_j\})$.
(2) $U$: the set of vertices who were ``refused'' to join $V_j$.
(3) $N$: the set of neighboring vertices to $V_j$ (who are not in $U$).

While $N$ is non-empty, the algorithm extracts an arbitrary vertex $v$ from $N$.
If $d_G(v,t_j)\le R(j)\cdot D(v)$ (the distance from $t_j$ to $v$ is at most $R_j$ times the distance from $v$ to its closest terminal),
then $v$ joins $V_j$. Otherwise $v$ joins $U$.
In the case where $v$ joins $V_j$, all its neighbors (outside of $U\cup V_j$), join $N$.
As each vertex might join $N$ at most once, eventually $N$ becomes empty. Then the procedure ceases and returns $V_j$.
\begin{algorithm}[!ht]
	\caption{$M=\texttt{Noisy-Voronoi}(G=(V,E,w),K=\{t_1,\dots,t_k\})$}\label{alg:mainSPR}
	\begin{algorithmic}[1]
		\STATE Set $\delta = \frac{1}{20\ln k}$ and $p=\frac15$. 
		\STATE Set $V_\perp~\la~V\setminus  K$.
			\hfill\emph{//~$V_\perp$ is the currently unclustered vertices.}					
		\FOR {$j$ from $1$ to $k$}
			\STATE Choose independently at random $g_j$ distributed according to $\Geo(p)$. \label{line:rnd}
			\STATE Set $R_j\la (1+\delta)^{g_j}$.	
			\STATE Set $V_j\la \texttt{Create-Cluster}(G,V_\perp ,t_j,R_j)$.
			\STATE Remove all the vertices in $V_j$ from $V_\perp$.
		\ENDFOR
		\RETURN the terminal-centered minor $M$ of $G$ induced by $V_1,\ldots,V_k$. \label{line:returnMinor}
	\end{algorithmic}	
\end{algorithm}
\begin{algorithm}[!ht]
	\caption{$V_j=\texttt{Create-Cluster}(G=(V,E,w),V_\perp ,t_j,R_j)$}\label{alg:CreateCluster}
	\begin{algorithmic}[1]
		\STATE Set $V_j\leftarrow \{t_j\}$.
		\STATE Set $U\la \emptyset$.
		\hfill\emph{//~$U$ is the set of vertices already denied from $V_j$.}					
		\STATE Set $N$ to be all the neighbors of $t_j$ in $V_\perp$.
		\WHILE {$N\ne \emptyset$}
		\STATE Let $v$ be an arbitrary vertex from $N$. \label{AlgLine:CreCluPickvN}
		\STATE Remove $v$ from $N$.
		\IF {$d_{G}(v,t_j)\le R_j\cdot D(v)$} \label{line:desidion}
		\STATE Add $v$ to $V_j$.
		\STATE Add all the neighbors of $v$ in $V_\perp\setminus \left(U\cup V_j\right)$ to $N$.  \label{line:AddNeighbors}
		\ELSE
		\STATE Add $v$ to $U$.
		\ENDIF
		\ENDWHILE
		\RETURN $V_j$.
	\end{algorithmic}	
\end{algorithm}

\begin{theorem}\label{thm:mainSPR}
	With probability $1-\frac{1}{k}$, in the minor graph $M$ returned by \Cref{alg:mainSPR}, it holds that for every two terminals $t,t'$, $d_{M}(t,t')\le O\left(\log k\right)\cdot d_{G}(t,t')$.
\end{theorem}

First we argue that \Cref{alg:mainSPR} indeed produces a terminal partition.
\begin{lemma}\label{lem:AlgRetPar}
	The sets $V_1,\ldots,V_k$ constructed by \Cref{alg:mainSPR} constitutes a terminal partition.
\end{lemma}
\begin{proof}
	It is straightforward from the description of the algorithm that the sets $V_1,\ldots,V_k$ are disjoint, and that for every $j$, $t_j\in V_j$ and $G[V_j]$ is connected.
	The only non trivial property we have to show is that every vertex $v\in V$ joins some cluster.

	Fix some $v\in V$, let $t_j$ be the closest terminal to $v$ (s.t. $D(v)=d_G(v,t_j)$), and let $P=\{t_j=u_0,u_1,\dots ,u_s=v\}$ be the shortest path from $t_j$ to $v$ in $G$.
	Note that as $P$ is a shortest path, $t_j$ is also the closest terminal to all the vertices in $P$.
	As $t_j=u_0\in V_j$, at least one vertex from $P$ is clustered during the algorithm.
	Let $u_{i'}$ be the first clustered vertex from $P$ (w.r.t time). 
	Denote by $V_{j'}$ the cluster $u_{i'}$ joins to.  We argue by induction on $i\ge i'$ that $u_{i}$ also joins $V_{j'}$. This will imply that $u_s=v$ joins $V_{j'}$ and thus is clustered.
	Suppose $u_i$ joins $V_{j'}$. It holds that  $d_{G}(u_i,t_{j'})\le R_{j'}\cdot D(u_i)$. 
	Moreover, all the neighbors of $u_i$ join $N$. Therefore $u_{i+1}$ necessarily joined to the set $N$ (at some stage during the execution of the \texttt{Create-Cluster} procedure for $V_{j'}$).
	As
	\begin{align*}
	d_{G}(u_{i+1},t_{j'}) & \le d_{G}(u_{i+1},u_{i})+d_{G}(u_{i},t_{j'})\\
	& \le d_{G}(u_{i+1},u_{i})+R_{j'}\cdot d_{G}(u_{i},t_{j})\\
	& \le R_{j'}\cdot d_{G}(u_{i+1},t_{j})= R_{j'}\cdot D(u_{i+1})~,
	\end{align*}
	$u_{i+1}$ will join $V_{j'}$, as required.
\end{proof}

\subsection{Modification}
Let $\hat{\Delta}=\min_{t,t'\in K}\{d_G(t,t')\}$ 
denote the minimal distance between a pair of terminals. Note that $\hat{\Delta}>0$.
For the sake of analysis we will make a preprocessing step to ensure that every edge $e$ has weight at most $c_w\cdot \hat{\Delta}=\frac{\delta}{24}\cdot \hat{\Delta}$.  
This can be achieved by subdividing larger edges, i.e. adding additional
vertices of degree two in the middle of such edges. 
Denote by $\hat{G}$ the modified graph $G$, when we repeatedly subdivide edges until every edge $e$ has small enough weight.
We argue that such subdivisions did not effect whatsoever the terminal-centered minor returned by \Cref{alg:mainSPR}.
\begin{claim}\label{clm:SameMinor}
	Let $G=(V,E,w)$ be a weighted graph with terminal set $K=\{t_1,\dots,t_k\}$. Consider an edge $e=\{v,u\}\in E$ of weight $\omega$. Let $\tilde{G}$ be the graph $G$ with subdivided edge $e$. Specifically, we add a new Steiner vertex $v_e$, and replace the edge $e$ by two new edges $\{v_e,v\},~\{v_e,u\}$, both of weight $\omega/2$.	
	
	Fix $g_1,\dots,g_k$ and consider \Cref{alg:mainSPR} where the random choices in \lineref{line:rnd} are $g_1,\dots,g_k$ respectively. Then the terminal-centered minor $M$ returned on input $G$ is the same as the terminal-centered minor $\tilde{M}$ returned on input $\tilde{G}$ .
\end{claim}
\begin{proof}
	As $g_1,\dots,g_k$ are fixed, \Cref{alg:mainSPR} is now deterministic.
	Let $V_1,\dots,V_k$ be the terminal partition induced by \Cref{alg:mainSPR} on $G$, and similarly let $\tilde{V}_1,\dots,\tilde{V}_k$ be the terminal partition induced by \Cref{alg:mainSPR} on $\tilde{G}$.
	We argue that for all $j$, $ V_j=\tilde{V}_j\setminus\{v_e\}$.
	Note that this will imply our claim. Indeed, let $V_j$, $V_{j'}$ be the clusters such that $v\in V_j$ and $u\in V_{j'}$. As each cluster is connected, necessarily  $v_e\in V_j\cup V_{j'}$. By the definition of subdivision, this will imply that the terminal-centered minors are indeed identical.
	
	Each Steiner vertex can be clustered only after at least one of its neighbors is clustered. Therefore $v_e$ cannot be clustered before both $v$ and $u$. 
	W.l.o.g $v$ joined $V_j$ while $u$ is still unclustered. 
	The vertex $v_e$ wasn't examined before the clustering of $v$. 
	Denote by $V_j'$ (resp. $\tilde{V}_j'$) the set $V_j$ (resp. $\tilde{V}_j$) right after the clustering of $v$ at the execution of \Cref{alg:mainSPR} on $G$ (resp. $\tilde{G}$).	
	Note that the order of extraction from $N$ in \lineref{AlgLine:CreCluPickvN} of \Cref{alg:CreateCluster} is determined deterministically. Therefore, up to the clustering of $v$ the algorithm behaved the same on both $G$ and $\tilde{G}$. 
	In particular, for all $j''<j$, $V_{j''}=\tilde{V}_{j''}$. Moreover, $V_j'=\tilde{V}_j'$.
	After $v$ joins $V_j$, $v_e$ joins (for the first time) to the set $N$ (for $\tilde{G}$).
	Note that
	\begin{align*}
	D(v_{e}) & =\min\left\{ D(v),D(u)\right\} +\frac{\omega}{2}\\
	d_{G}(t_{j,}v_{e}) & =\min\left\{ d_{G}(t_{j,}v),d_{G}(t_{j},u)\right\} +\frac{\omega}{2}
	\end{align*}
	As $v$ joined $V_j$, necessarily $d_G(t_j,v)\le R_j\cdot D(v)$. Consider the following cases:
	\begin{itemize}
		\item \textbf{$u\notin V_j$ : } In the algorithm for $G$, $u$ was examined (as $v\in V_j$), thus $d_G(t_j,u)>R_j\cdot D(u)$. Therefore $u$ will also not join $\tilde{V}_j$.
		As $v_e$ has edges only to $v$ and $u$, $v_e$ has no impact on any other vertex. Therefore the cluster $\tilde{V}_j$ will be constructed in the same manner as $V_j$ (up to maybe containing $v_e$).
		Note that all the other clusters will not be effected, as if $v_e$ remained unclustered, it becomes a leaf. We conclude that for every $j''$, $V_{j''}=\tilde{V}_{j''}\setminus\{v_e\}$. 	
		
		\item \textbf{$u\in V_j$ : } It holds that $d_G(t_j,u)\le R_j\cdot D(u)$. Therefore 
		\[
		d_{G}(t_{j,}v_{e})=\min\left\{ d_{G}(t_{j,}v),d_{G}(t_{j,}e)\right\} +\frac{\omega}{2}\le R_{j}\cdot\min\left\{ D(v),D(u)\right\} +\frac{\omega}{2}\le R_{j}\cdot D(v_{e})~.
		\]
		Therefore $v_e$ will join $\tilde{V}_j$, which will ensure that $u$ joins $\tilde{N}$, and afterwards to $\tilde{V}_j$. Note that $v_e$ has no other impact. In particular, for every $j''\ne j$, $V_{j''}=\tilde{V}_{j''}$ while $V_{j}\cup\{v_e\}=\tilde{V}_{j}$.
	\end{itemize} 
\end{proof}

Consider the modified graph $\hat{G}$. Suppose that we proved that with probability at least $1-\frac{1}{k}$, in the minor graph $\hat{M}$ returned by \Cref{alg:mainSPR} for $\hat{G}$, it holds that for every two terminals $t,t'$, $d_{\hat{M}}(t,t')\le O\left(\log k\right)\cdot d_{\hat{G}}(t,t')= O\left(\log k\right)\cdot d_{G}(t,t')$. Then by repetitive use of \Cref{clm:SameMinor} (once for every new vertex), \Cref{thm:mainSPR} follows.
From now on, we will abuse notation and refer to the graph $\hat{G}$ as $G$. Note that all this is done purely for the sake of analysis, as by \Cref{clm:SameMinor} we will get the same minor when running \Cref{alg:mainSPR} for either $G$ or $\hat{G}$. Thus, in fact, we will execute \Cref{alg:mainSPR} on the original graph with no modifications.

\section{Distortion Analysis}\label{sec:distortion}
\subsection{Interval and Charges}
In this section we describe in detail the probabilistic process of breaking the graph into clusters from the view point of the Steiner vertices. The main objective will be to define a charging scheme, which we can later use to bound the distortion.

Consider two terminals
$t$ and $t'$.
Let $P_{t,t'}=\left\{ t=v_{0},\dots,v_{\gamma}=t'\right\}$ be the shortest path from $t$ to $t'$ in $G$.
We can assume that there are no terminals in $P_{t,t'}$ other than $t,t'$. This is because if we will prove that for every pair of terminals  $t,t'$ such that $P_{t,t'}\cap K=\{t,t'\}$ it holds that $d_{M}(t,t')\le O(\log k)\cdot d_{G}(t,t')$, the this property will be implied for all terminal pairs.

For an interval
$Q=\left\{ v_{a},\dots,v_{b}\right\} \subseteq P_{t,t'}$, the \emph{internal length} is 
$L(Q)=d_{G}(v_{a},v_{b})$, while the \emph{external length} is $L^{+}(Q)=d_{G}(v_{a-1},v_{b+1})$ \footnote{For ease of notation we will denote $v_{-1}=t$ and $v_{\gamma+1}=t'$.}.
The distance from the interval $Q$ to the terminals, denoted $D(Q)=D(v_a)$ is simply the distance from its leftmost point $v_a$ to the closest terminal to $v_a$.
Set $\cint=\frac16$ (``int'' for interval).
We partition the vertices in $P_{t,t'}$ into consecutive intervals $\mathcal{Q}$, such that for every $Q\in \mathcal{Q}$,
\begin{eqnarray}
L(Q)\le\cint\delta\cdot D(Q)\le L^{+}(Q)~.\label{eq:IntervalLenght}
\end{eqnarray}
Such a partition could be constructed as follows: Sweep
along the interval $P_{t,t'}$ in a greedy manner, after partitioning the prefix $v_{0},\dots,v_{h-1}$,
to construct the next $Q$, simply pick
the minimal index $s$ such that $ L^{+}(\left\{ v_{h},\dots,v_{h+s}\right\} )\ge\cint\delta\cdot D(v_{h})$.
By the minimality of $s$, $L(\left\{ v_{h},\dots,v_{h+s}\right\} )\le L^{+}(\left\{ v_{h},\dots,v_{h+s-1}\right\} )\le\cint\delta\cdot D(v_{h})$ (in the case $s=0$, trivially $L(\left\{ v_{h}\right\} )=0\le\cint\delta\cdot D(v_{h})$).
Note that such $s$ could always be found, as $ L^{+}(\left\{ v_{h},\dots,v_{\gamma}\right\} )=d_{G}(v_{h-1},t')\ge d_{G}(v_{h},t')\ge D(v_{h})=D(Q)$.

In the beginning of \Cref{alg:mainSPR}, all the vertices of $P_{t,t'}$ are
\emph{active}. 
Consider round $j$ in the algorithm when terminal $t_{j}$ constructs its cluster $V_{j}$. Specifically, it picks $g_{j}$
and sets $\ensuremath{R_{j}\leftarrow (1+\delta)^{g_j}}$. Then, using the \texttt{Create-Cluster} procedure it grows a cluster in a ``Dijkstra''  fashion.
If no active vertex joins $V_j$, we say that $t_j$ doesn't \emph{participate} in $P_{t,t'}$. 
Otherwise, let $a_{j}\in P_{t,t'}$
(resp., $b_{j}$) be the active vertex that joins to $V_{j}$ with minimal (resp., maximal) 
index (w.r.t $P_{t,t'}$).
All the vertices
$\left\{ a_{j},\dots,b_{j}\right\}\subset P_{t,t'}$ between $a_{j}$ and $ b_{j}$ (w.r.t the order induced by $P_{t,t'}$)
become inactive. We call this set $\left\{ a_{j},\dots,b_{j}\right\}$ a
\emph{detour} $\mathcal{D}_{j}$ from $a_{j}$ to $b_{j}$.
See \Cref{fig:sliceChange} for an illustration.
\begin{figure}[]
	\centering{\includegraphics[scale=0.7]{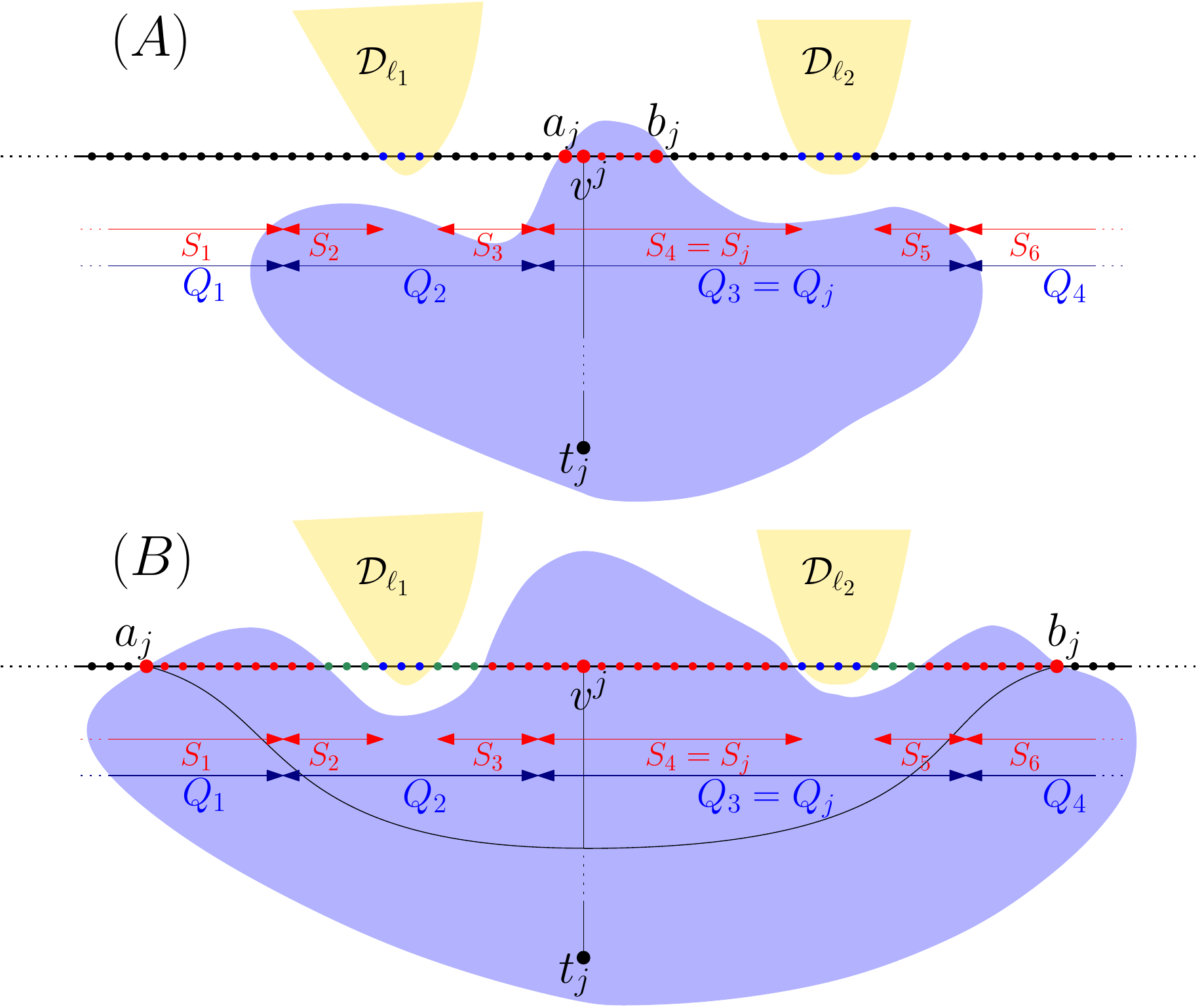}} 
	\caption{\label{fig:sliceChange}\small 
		\it 
		The figure illustrates round $j$ in \Cref{alg:mainSPR}, when $t_j$ grows the cluster $V_j$.
		We present two scenarios for different choices of $R_j$.
		The black line is part of $P_{t,t'}$ the shortest path from $t$ to $t'$. 
		The blue intervals $Q_i$ represent the intervals in $\mathcal{Q}$. The red sub-intervals $S_i$ represent the slices (maximal continuous
		subsets of active vertices). Where $S_2,S_3\subset Q_2$ and $S_4,S_5\subset Q_3$.
		The yellow areas represent detours $\mathcal{D}_{\ell_1}$ and $\mathcal{D}_{\ell_2}$, where $Q_2$ (resp., $Q_3$) is charged for $\mathcal{D}_{\ell_1}$ (resp., $\mathcal{D}_{\ell_2}$). Note that vertices in that areas are inactive.\newline
		The terminal $t_j$ increases gradually $R_j$, the first vertex to be covered is $v^j$. In scenario (A), the growth of $R_j$ terminates immediately after covering $v^j$, and sets the borderline vertices $a_j$ and $b_j$ within the subinterval $S_j$. While in scenario (B), the growth of $R_j$ continues for another step, setting both $a_j$ and $b_j$ out of $S_j$.
		Vertices already inactive are colored in blue. 
		Vertices who join the cluster $V_j$ are colored in red.
		The green vertices, are vertices which still un-covered, but nevertheless become inactive. Vertices which remain active after the creation of $V_j$, are colored in black. \newline
		In scenario (A) all the vertices that become inactive, $\mathcal{D}_j$, are included in $S_4$. 
		$Q_3$ is charged for $\mathcal{D}_j$. The number of slices in $Q_3$ is increased by $1$, and no other changes occur ($X(Q_2)=1$, $X(Q_3)=2$).
		In scenario (B) $\mathcal{D}_\ell$ contains all the vertices in $S_2,S_3,S_4,S_5$, and part of the vertices in $S_1,S_6$.
		The number of slices in $Q_2$ and $Q_3$ become $0$, while the number of slices in $Q_1$ and $Q_4$ remain unchanged.
		$Q_3$ is charged for $\mathcal{D}_\ell$, while its charge for $\mathcal{D}_{\ell_2}$ is erased. Additionally, the charge of $Q_2$ for $\mathcal{D}_{\ell_1}$ is erased. That is, $Q_2$ will remain uncharged till the end of the algorithm ($\tilde{X}(Q_2)=X(Q_2)=0$, $X(Q_3)=1$).
	}
\end{figure}

Within each interval $Q$, each maximal sub-interval of active
vertices is called a \emph{slice}. We denote by $\mathcal{S}(Q)$ the current
number of slices in $Q$. In the beginning of the algorithm, for every
interval $Q$, $\mathcal{S}(Q)=1$, while at the end of the algorithm $\mathcal{S}(Q)=0$.

For an active vertex $v$, let $r_{v}$ be the minimal choice of $R_{j}$ (determined by $g_j$),
that will force $v$ to join $V_{j}$. Let $v^{j}$ be the
active vertex with minimal $r_{v}$ (breaking ties arbitrarily).
Note that $V_j$ is monotone with respect to $R_j$. That is, if $v$ will join $V_j$ for $R_j=r$, it will join $V_j$ for $R_j=r'\ge r$ as well.
We denote by $Q_{j}\in\mathcal{Q}$ the interval containing $v^{j}$. 
Similarly, $S_{j}$ is the slice containing $v^{j}$.
We \emph{charge}
$Q_{j}$
for the detour $\mathcal{D}_{j}$. We denote by $X(Q)$ the number of
detours the interval $Q$ is currently charged for. For every detour
$\mathcal{D}_{j'}$ which is contained in $\mathcal{D}_{j}$ (that is $a_{j}<a_{j'}<b_{j'}<b_{j}$ w.r.t. the order induced by $P_{t,t'}$),
we erase the detour and its charge. That is, for every $Q'\ne Q_{j}$,
$X(Q')$ might only decrease, while $X(Q_j)$ might increase
by at most $1$ (and can also decrease as a result of deleted detours).
We denote by $\tilde{X}(Q)$ the size of $X(Q)$ by the end of \Cref{alg:mainSPR}.
\Cref{fig:sliceChange} illustrates a single step.

Next, we analyze the change in the number of slices as a result of constructing the cluster $V_j$.
If $R_{j}<r_{v^j}$, then no active vertex joins $V_{j}$
and therefore $X(Q)$ and $\mathcal{S}(Q)$ stay unchanged, for all $Q\in\mathcal{Q}$. Otherwise, $R_{j}\ge r_{v^j}$,
a new detour will appear, and will be charged upon $Q_{j}$.
All the slices $S$ which are contained in $\mathcal{D}_{j}$ are deleted.
Every slice $S$ that intersects $\mathcal{D}_{j}$ but is not contained
in it will be replaced by one or two new slices. If
$\mathcal{D}_{j}\cap S \notin\{\mathcal{D}_{j},S\}$,
then $S$ is replaced by a single new sub-slice $S'$. The only possibility
for a slice to be replaced by two sub-slices is if $\mathcal{D}_{j}\subseteq S$,
and $\mathcal{D}_{j}$ does not contain an  ``extremal'' vertex in $S$ (see \Cref{fig:sliceChange}, scenario (A)). This can
happen only at $S_{j}$. We conclude that for every $Q'\ne Q_{j}$,
$\mathcal{S}(Q')$ might only decrease, while $\mathcal{S}(Q_{j})$ might increase
by at most $1$.

\begin{claim}\label{clm:failProb}
	Assuming $R_j\ge r_{v^j}$, all of $S_j$ joins $V_j$ with probability at least $1-p$.
\end{claim}
\begin{proof}
	As $v^j$ joins $V_j$ for $R_j\ge r_{v^j}$, by \lineref{line:desidion} of \Cref{alg:CreateCluster}, necessarily $\frac{d_{G}(v^{j},t_{j})}{D(v^{j})}\le r_{v^j}$. We will argue that for every $u\in S_j$, the following inequality holds:
	\begin{eqnarray}
	\frac{d_{G}(u,t_{j})}{D(u)}\le\frac{d_{G}(v^{j},t_{j})}{D(v^{j})}\left(1+\delta\right)\le r_{v^{j}}\left(1+\delta\right)~.\label{eq:SmallerRj}
	\end{eqnarray}
	Next, assume that $R_{j}\ge(1+\delta)r_{v^{j}}$.
	Before the execution of the \texttt{Create-Cluster} procedure for $V_j$, all the vertices in $S_j$ belong to $V_\perp$ (as all of them are active). Because $R_j\ge r_{v^j}$, $v^j$ will join $V_j$ (by the definition of $r_{v^j}$). In particular, additional vertices from $S_j$ (if exist) will join $N$. Using  inequality (\ref{eq:SmallerRj}), for every $u\in S_j$, $d_{G}(u,t_{j})/D_u\le r_{v_j}(1+\delta)\le R_j$. 
	Therefore every vertex from $S_j$ joining $N$ will also join $V_j$.
	In such a way, since $S_j$ is connected in $V_\perp$, all the vertices of $S_j$ will join $V_j$, as required.
	
	Next, we analyze the probability that indeed $R_j\ge (1+\delta) r_{v^j}$. Recall that $R_j=(1+\delta)^{g_j}$ where $g_j$ is distributed according to geometric distribution with parameter $P_{t,t'}$. Conditioned on the event $R_j\ge r_{v^j}$, we have that

	\begin{eqnarray}
		\Pr\left[R_{j}\ge(1+\delta)r_{v^{j}}\mid R_{j}\ge r_{v^{j}}\right] & =&\Pr\left[g_{j}\ge\log_{1+\delta}\left((1+\delta)r_{v^{j}}\right)\mid g_{j}\ge\log_{1+\delta}r_{v^{j}}\right]\\
		&=&\Pr\left[g_{j}\ge1+\log_{1+\delta}r_{v^{j}}\mid g_{j}\ge\log_{1+\delta}r_{v^{j}}\right]=1-p~.\label{eq:RjBigProb}
	\end{eqnarray}
It remains to prove inequality (\ref{eq:SmallerRj}).
	By the definition of $D(Q_j)$ and the triangle inequality 
	\begin{eqnarray}
		L(Q_{j})\overset{(\ref{eq:IntervalLenght})}{\le}\cint\delta\cdot D(Q_{j})\le\cint\delta\cdot\left(D(v^{j})+L(Q_{j})\right)\le2\cint\delta\cdot D(v^{j})\le2\cint\delta\cdot d_G(v^j,t_j)~.\label{eq:BoundLQj}
	\end{eqnarray}
	Therefore, for every $u\in S_{j}$,
	\[
	d_{G}(u,t_{j})\le d_{G}(v^{j},t_{j})+L(Q_{j})\overset{(\ref{eq:BoundLQj})}{\le} d_{G}(v^{j},t_{j})\left(1+2\cint\delta\right)~,
	\]
	Similarly,
	\begin{eqnarray}
		D(u)\ge D(v^{j})-L(Q_{j})\ge D(v^{j})\left(1-2\cint\delta\right)~.\label{eq:BoundDu}
	\end{eqnarray}
	We conclude that 
	\[
	\frac{d_{G}(u,t_{j})}{D(u)}\le\frac{d_{G}(v^{j},t_{j})\left(1+2\cint\delta\right)}{D(v^{j})\left(1-2\cint\delta\right)}\le\frac{d_{G}(v^{j},t_{j})}{D(v^{j})}\left(1+3\cdot2\cint\delta\right)=\frac{d_{G}(v^{j},t_{j})}{D(v^{j})}\left(1+\delta\right)~.
	\]
\end{proof}

\subsection{Bounding the Number of Failures}\label{subsec:FailCount}
We define a \emph{cost function} $f:\mathbb{R}_+^{|\mathcal{Q}|}\rightarrow\mathbb{R}_{+}$,
in the following way $f(\{x_Q\}_{Q\in\mathcal{Q}})=\sum_{Q\in\mathcal{Q}}X(Q)\cdot L^{+}(Q)$ .\footnote{Even though our goal will be to bound $f(\{x_Q\}_{Q\in\mathcal{Q}})$, we define $f$ as a general function from $\mathbb{R}^{|\mathcal{Q}|}$ in order to use it on other variables as well.} 
Note that the cost function $f$ is linear and monotonically increasing
coordinate-wise.
In \Cref{sec:DistBound} we show that the distance  $d_M(t,t')$ between $t$ and $t'$ in the minor graph $M$ can be bounded by $\log k\cdot f\left(\{\tilde{X}(Q)\}_{Q\in\mathcal{Q}}\right)$, the scaled cost function applied on the charges.
This section is devoted to proving the following lemma.
\begin{lemma}\label{lem:fbound}
	$\Pr\left[f\left(\{\tilde{X}(Q)\}_{Q\in\mathcal{Q}}\right)\ge43\cdot d_G(t,t') \right]\le k^{-3}$.
\end{lemma}
Using \Cref{clm:failProb}, one can show that for every $Q\in\mathcal{Q}$, $\mathbb{E}[\tilde{X}(Q)]=O(1)$, and moreover, w.h.p. $\tilde{X}(Q)=O(\log k)$ for all $Q$. 
However, we use a concentration bound on all $\{\tilde{X}(Q)\}_{Q\in\mathcal{Q}}$ simultaneously in order to provide a stronger upper bound.

\subsubsection{Bounding by independent variables}
	In our journey to bound  $f\left(\{\tilde{X}(Q)\}_{Q\in\mathcal{Q}}\right)$, the first step will be to replace $\{\tilde{X}(Q)\}_{Q\in\mathcal{Q}}$ with independent random variables. 
	Consider the following process: a \emph{box} $B$ which contains coins of two types: active and inactive. 
	In the beginning, there is a single active
	coin. In each round, we toss an active coin, which gets $0$ (failure) with probability $p$, and $1$ (success) with probability $1-p$. If we get a $0$, two additional active coins are added
	to the box. 
	In any case, the tossed coin becomes inactive. All the coin tosses throughout the proses are independent.
	The process terminates when no active coins remain.
	Let $\{B_Q\}_{Q\in\mathcal{Q}}$ be a set of $|\mathcal{Q}|$ independent boxes (here the box $B_Q$ resembles the interval $Q$). 
	For the box $B_Q$, denote by $Z(Q)$ the number of active coins,
	by $Y(Q)$ the number of inactive coins and by $\tilde{Y}(Q)$ the
	number of inactive coins at the end of the process. 
\begin{claim}
	\label{clm:CoinsDominate}For every $\alpha\in\mathbb{R}_{+}$, $\Pr\left[f\left(\{\tilde{X}(Q)\}_{Q\in\mathcal{Q}}\right)\ge\alpha\right]\le\Pr\left[f\left(\{\tilde{Y}(Q)\}_{Q\in\mathcal{Q}}\right)\ge \alpha\right]$.
\end{claim}
\begin{proof}
	The proof is done by coupling the two processes of \Cref{alg:mainSPR} and the coin tosses. We execute \Cref{alg:mainSPR}, which implicitly induces slices and detour charges. Simultaneously, we will use \Cref{alg:mainSPR} to toss coins.
	Inductively, we will maintain the invariant that
	$\{Y(Q)\}_{Q\in\mathcal{Q}}$ and $\{Z(Q)\}_{Q\in\mathcal{Q}}$
	are no less then $\{X(Q)\}_{Q\in\mathcal{Q}}$ and $\{S(Q)\}_{Q\in\mathcal{Q}}$
	(respectively) coordinate-wise.
	
	In the beginning 
	$\{X(Q)\}_{Q\in\mathcal{Q}}=\{Y(Q)\}_{Q\in\mathcal{Q}}=\left\{0\right\}_{Q\in\mathcal{Q}}$
	and $\{S(Q)\}_{Q\in\mathcal{Q}}=\{Z(Q)\}_{Q\in\mathcal{Q}}=\left\{1\right\}_{Q\in\mathcal{Q}}$.
	Consider round $j$, where the cluster $V_j$ is created for the terminal $t_j$. If $R_{j}<r_{v^j}$
	then nothing happens, and the invariant holds. Else, $R_{j}\ge r_{v^j}$,
	we will make a coin toss from the $B_{Q_{j}}$ box. Let $p'$
	be the probability that not all of $S_j$ joins $V_j$.
	By \Cref{clm:failProb},  $p'\le p$. 
	If indeed not all of $S_{j}$ joins $V_j$,
	the toss result is set to $0$. Otherwise, with probability $\frac{p-p'}{1-p'}$ the toss set to $0$.
	Note that the probability of $0$ is exactly $p'\cdot 1+(1-p')\cdot\frac{p-p'}{1-p'}=p$. 
	
	Next we argue that the invariant is maintained in either case.
	If not all of $S_j$ joins $Q_j$, then $S(Q_j)$ might increase by at most one, while the number of active coins $Z_{Q_j}$ increases by exactly one. 
	Otherwise, all of $S_j$ joins $Q_j$. In this case  $S(Q_j)$ necessarily decreases by at least one, while $Z_{Q_j}$ might either decrease or increase by one.
	For the charge parameter, $X(Q_j)$ might increase by at most one, while the number of inactive coins $Y(Q_j)$ increases by exactly one.
	For every $Q'\ne Q_{j}$, $\mathcal{S}(Q')$
	and $X(Q')$ might only decrease, while $Z_{Q'}$ and $Y(Q')$
	stay unchanged. We conclude that the invariant is holds after the construction of the cluster $V_j$. 
	
	At the end of the algorithm (when no slices are left), we might still
	have some active coins. In this case we will simply toss coins until no active coins remain (note that this indeed happens with probability $1$). Note that by
	doing so $\{Y(Q)\}_{Q\in\mathcal{Q}}$ can only
	grow coordinate-wise. As the marginal distribution on $\{\tilde{Y}(Q)\}_{Q\in\mathcal{Q}}$
	is exactly identical to the original one, the claim follows.
\end{proof}

\subsubsection{Replacing Coins with Exponential Random Variables}
Our next step is to replace each $Y(Q)$ with exponential random variable. This replacement will make the use of concentration bounds more convenient.
Consider some box $B_{Q}$. An equivalent way to describe the probabilistic
process in $B_{Q}$ is the following. Take a single coin with failure
probability $p$, toss this coin until the number of successes exceeds
the number of failures. The total number of tosses is exactly $\tilde{Y}(Q)$.
Note that $\tilde{Y}(Q)$ is necessarily odd. Next we bound the probability
that $\tilde{Y}(Q)\ge2m+1$, for $m\ge1$. This is obviously upper
bounded by the probability that in a series of $2m$ tosses we had
at least $m$ failures (as otherwise the process would have stopped earlier,
in fact this true even for $2m-1$ tosses). Let $\chi_{i}$ be an indicator
for a failure in the $i$'th toss, and  $\chi=\sum_{i=1}^{2m}\chi_{i}$. Note that
$\mathbb{E}\left[\chi\right]=2m\cdot p$.
A bound on $\chi$ follows by Chernoff inequality.
\begin{fact}[Chernoff inequality]\label{fact:Chernoff}
	Let $X_{1},\dots,X_{n}$ be i.i.d indicator variables each with probability
	$p$. Set $X=\sum_i X_{i}$ and $\mu=\mathbb{E}[X]=np$. Then for every
	$\delta\le2e-1$, $\Pr\left[X\ge(1+\delta)\mu\right]\le\exp(-\mu\delta^{2}/4)$.
\end{fact}
\begin{align*}
\Pr\left[\tilde{Y}(Q)\ge2m+1\right]\le\Pr\left[\chi\ge m\right] & =\Pr\left[\chi\ge\left(1+(\frac{1}{2p}-1)\right)\mathbb{E}[\chi]\right]\\
& \le\exp\left(-2m\cdot p\cdot(\frac{1}{2p}-1)^{2}/4\right)=\exp\left(-\frac{9}{40}m\right)\le\exp\left(-\frac{1}{5}m\right)~.
\end{align*}
\sloppy We conclude that the distribution of $\tilde{Y}(Q)$ is dominated
by $1+\Exp\left(10\right)$ (as for $W\sim\Exp(10)$, $\Pr\left[1+W\ge2m+1\right]=\exp\left(-\frac m5\right)$). Let $(\{W(Q)\}_{Q\in\mathcal{Q}})$ be i.i.d. random variables distributed according to $\Exp(10)$,
 since all the boxes are independent and $f$ is linear and monotone
coordinate-wise, we conclude:
\begin{claim}
	\label{clm:ExpDominate}For every 	$\alpha\in\mathbb{R}_{+}$,\\ \phantom{.\hspace{70pt}}$\Pr\left[f\left(\left\{ \tilde{Y}(Q)\right\} _{Q\in\mathcal{Q}}\right)\ge\alpha\right]\le\Pr\left[f\left(\left\{ 1\right\} _{Q\in\mathcal{Q}}\right)+f\left(\left\{ W(Q)\right\} _{Q\in\mathcal{Q}}\right)\ge\alpha\right]$.
\end{claim}
\begin{proof}
	Set $\varphi=|\mathcal{Q}|$. Let $Q^1,Q^2,\dots,Q^{\varphi}$ be some arbitrarily fixed ordering of the intervals.
	For $s\in [\varphi]$, set $f_{\setminus\{s\}}(x_{1},\dots,x_{s-1},x_{s+1},\dots,x_{\varphi})=\sum_{i\in[\varphi]\setminus\left\{ s\right\} }x_{i}\cdot L^{+}(Q^{i})$.
	When integrating over the appropriate measure space, it holds that
	\begin{align*}
	\Pr\left[f\left(\tilde{Y}(Q^{1}),\dots,\tilde{Y}(Q^{\varphi})\right)\ge\alpha\right] & =\int_{\beta}\Pr\left[f_{\setminus\{1\}}\left(\tilde{Y}(Q^{2}),\dots,\tilde{Y}(Q^{\varphi})\right)=\beta\right]\\
	& \qquad\qquad\qquad\qquad\cdot\Pr\left[\tilde{Y}(Q^{1})\cdot L^{+}(Q^{1})\ge\alpha-\beta\right]d\beta\\
	& \le\int_{\beta}\Pr\left[f_{\setminus\{1\}}\left(\tilde{Y}(Q^{2}),\dots,\tilde{Y}(Q^{\varphi})\right)=\beta\right]\\
	& \qquad\qquad\qquad\qquad\cdot\Pr\left[\left(1+W(Q^{1})\right)\cdot L^{+}(Q^{1})\ge\alpha-\beta\right]d\beta\\
	& =\Pr\left[f\left(1+W(Q^{1}),\tilde{Y}(Q^{2}),\dots,\tilde{Y}(Q^{\varphi})\right)\ge\alpha\right]\\
	& \le\Pr\left[f\left(1+W(Q^{1}),1+W(Q^{2}),\tilde{Y}(Q^{3}),\dots,\tilde{Y}(Q^{\varphi})\right)\ge\alpha\right]\\
	& \le\cdots\le\Pr\left[f\left(1+W(Q^{1}),\dots,1+W(Q^{\varphi})\right)\ge\alpha\right]\\
	& =\Pr\left[f\left(1,\dots,1\right)+f\left(W(Q^{1}),\dots,W(Q^{\varphi})\right)\ge\alpha\right]~.
	\end{align*}
\end{proof}

\subsubsection{Concentration}
Set $\Delta=d_{G}(t,t')$. It holds that 
\begin{equation*}
\Delta\le\sum_{Q\in\mathcal{Q}}L^{+}(Q)\le2\Delta~,
\end{equation*}
as every edge in $P_{t,t'}$ is counted at least once, and at most
twice in this sum. In particular $f\left(\left\{ 1\right\} _{Q\in\mathcal{Q}}\right)\le2\Delta$.
Recall that by our modification step, every edge in $P_{t,t'}$ is of weight at most $c_w\cdot \Delta$.
In particular, for every $Q\in\mathcal{Q}$, $L^{+}(\mathcal{Q})\le L(\mathcal{Q})+2c_{w}\cdot\Delta$.
For every vertex $v$ on $P_{t,t'}$, it holds that $D(v)\le\min\left\{ d_{G}(v,t),d_{G}(v,t')\right\} \le\frac{\Delta}{2}$.
Therefore for every $Q\in\mathcal{Q}$, 
\[
L^{+}(\mathcal{Q})\le L(\mathcal{Q})+2c_{w}\cdot\Delta\overset{\eqref{eq:IntervalLenght}}{\le}\cint\delta\cdot D(Q)+2c_{w}\cdot\Delta\le\left(\frac{\cint\delta}{2}+2c_{w}\right)\cdot\Delta=\cint\delta\cdot\Delta~.
\]

Let $\tilde{W}(Q)\sim L^{+}(Q)\cdot\Exp\left(10\right)$.
In particular, $\tilde{W}(Q)\sim\Exp\left(10\cdot L^{+}(Q)\right)$.
Set $\tilde{W}=\sum_{Q\in\mathcal{Q}}\tilde{W}(Q)$. Then $f\left(\left\{ W(Q)\right\} _{Q\in\mathcal{Q}}\right)$
is distributed exactly as $\tilde{W}$. The maximal mean among the $\tilde{W}(Q)$'s
is $\lambda_{M}=\max_{Q\in\mathcal{Q}}10\cdot L^{+}(Q)\le10\cdot\cint\delta\cdot\Delta$.
The mean of $\tilde{W}$ is $\mu=\sum_{Q\in\mathcal{Q}}10\cdot L^{+}(Q)\le20\Delta$.
Set $c_{\text{con}}=\frac12$  
(con for concentration). Using \Cref{clm:CoinsDominate}, \Cref{clm:ExpDominate} and \Cref{lem:ExpConcentration},
we conclude 

\begin{align*}
\Pr\left[f\left(\left\{ \tilde{X}(Q)\right\} _{Q\in\mathcal{Q}}\right)\ge(c_{\text{con}}+42)\Delta\right] & \le\Pr\left[f\left(\left\{ \tilde{Y}(Q)\right\} _{Q\in\mathcal{Q}}\right)\ge(c_{\text{con}}+42)\Delta\right]\\
& \le\Pr\left[f\left(\left\{ W(Q)\right\} _{Q\in\mathcal{Q}}\right)\ge(c_{\text{con}}+42)\Delta-f\left(\left\{ 1\right\} _{Q\in\mathcal{Q}}\right)\right]\\
& \le\Pr\left[\tilde{W}\ge(c_{\text{con}}+40)\Delta\right]\\
& \le\exp\left(-\frac{1}{2\lambda_{M}}\left(\left(c_{\text{con}}+40\right)\Delta-2\mu\right)\right)\\
& \le\exp\left(-\frac{1}{2}\cdot\frac{1}{10\cint\delta\Delta}\cdot c_{\text{con}}\Delta\right)=\exp\left(-\frac{c_{\text{con}}}{20\cdot\cint\delta}\right)=k^{-3}~.
\end{align*}
Note that $c_{\text{con}}\le1$, thus \Cref{lem:fbound} follows.

\subsection{Bounding the Distortion}\label{sec:DistBound}
Denote by $\EfBig$ the event that for some pair of terminals $t,t'$, 
$f\left(\{\tilde{X}(Q)\}_{Q\in\mathcal{Q}}\right)\ge43 \cdot d_G(t,t')$. \footnote{We abuse notation here and use the same $\{\tilde{X}(Q)\}_{Q\in\mathcal{Q}}$ for all terminals.}
By \Cref{lem:fbound} and the union bound, $\Pr\left[\EfBig\right]\le {k \choose 2}\cdot  k^{-3}< \frac{1}{2k}$.

Let $\EB$ be the event that for some $j$, $R_j> c_d$, where $c_d=e^2$.
Note that if $\EB$ does not hold, then every vertex $v$ joins to a cluster $V_j$ such that $d_G(v,t_j)\le c_d\cdot D(v)$.
\begin{claim}\label{clm:boundEB}
	$\Pr[\EB]\le \frac{1}{2k}$.
\end{claim}
\begin{proof}
	Let $\EB_j$ be the event that  $R_j>c_d$.	It holds that
	$$\Pr[\EB_j]=\Pr[g_j\ge\log_{1+\delta}c_d]\le (1-p)^{\log_{1+\delta}c_d-1}
	\le (1-p)^{\frac2\delta-1}
	\le\frac{1}{k^3}~,$$
	where the second inequality holds as $\log_{1+\delta}c_{d}=\frac{\ln c_{d}}{\ln1+\delta}\ge\frac{2}{\delta}$.
	By the union bound, $\Pr[\EB]\le\frac{1}{k^2}\le \frac{1}{2k}$ as required.
\end{proof}

\begin{lemma}\label{lem:distortion}
	Assuming $\overline{\EB}$ and $\overline{\EfBig}$, for every pair of terminals $t,t'$, $d_{M}(t,t')\le O(\log k)\cdot d_{G}(t,t')$.
\end{lemma}
\begin{proof}
Fix some $t,t'$. 
By the end of \Cref{alg:mainSPR}, all the vertices in $P_{t,t'}=\left\{ t=v_{0},\dots,v_{\gamma}=t'\right\} $
are divided into consecutive detours \footnote{Note that we consider only detours who inflict a charge by the end of the algorithm. Therefore the detours are disjoint and every vertex in $P_{t,t'}$ belongs to some detour.} $\mathcal{D}_{\ell_1},\dots,\mathcal{D}_{\ell_{k'}}$. The detour
$\mathcal{D}_{\ell_j}$ was constructed at round $\ell_{j}$ by the terminal $t_{\ell_j}$. 
The detour $\mathcal{D}_{\ell_j}$ was charged upon the interval $Q_{\ell_j}$, which contains the vertex $v^{\ell_j}$.
The leftmost vertex in $\mathcal{D}_{\ell_{j}}$ is called
$a_{\ell_{j}}$, while the rightmost vertex is called $b_{\ell_{j}}$. In particular, for every $1\le j\le k'-1$, there is an edge in $G$ between $b_{\ell_{j}}$ and $a_{\ell_{j+1}}$, and therefore there is an edge between $t_{\ell_{j}}$ to $t_{\ell_{j+1}}$ in the terminal-centered minor $M$. 
As $t=v_0$ joins the cluster of itself, necessarily $t_{\ell_1}=t$. Similarly $t_{\ell_{k'}}=t'$.
See \Cref{fig:distortion} for an illustration.
Using the triangle inequality, we conclude,
\begin{align*}
d_{M}(t,t')\le\sum_{j=1}^{k'-1}d_{G}(t_{\ell_{j}},t_{\ell_{j+1}}) & \le\sum_{j=1}^{k'-1}\left[d_{G}(t_{\ell_{j}},v^{\ell_{j}})+d_{G}(v^{\ell_{j}},v^{\ell_{j+1}})+d_{G}(v^{\ell_{j+1}},t_{\ell_{j+1}})\right]\\
& \le\sum_{j=1}^{k'-1}d_{G}(v^{\ell_{j}},v^{\ell_{j+1}})+2\sum_{j=1}^{k'}d_{G}(t_{\ell_{j}},v^{\ell_{j}})\\
& \le d_{G}(t,t')+2\sum_{j=1}^{k'}c_{d}\cdot D(v^{\ell_{j}})
\end{align*}
where the last inequality follows by our assumption $\overline{\EB}$.
By the definition of $D(Q_{\ell_j})$, inequality \eqref{eq:IntervalLenght} and triangle inequality,
$D(v^{\ell_{j}})\le D(Q_{\ell_{j}})+L(Q_{\ell_{j}})\le\left(\frac{1}{\cint\delta}+1\right)L^{+}(Q_{\ell_{j}})\le\frac{2}{\cint\delta}\cdot L^{+}(Q_{\ell_{j}})$.
Using the assumption $\overline{\EfBig}$, we conclude,

\begin{eqnarray}
	\hspace{100pt}d_{M}(t,t') & \le& d_{G}(t,t')+2c_{d}\sum_{i=1}^{k'}\frac{2}{\cint\delta}\cdot L^{+}(Q_{\ell_{i}})\label{eq:distBound2}\\
	& =&d_{G}(t,t')+\frac{4c_{d}}{\cint\delta}\sum_{Q\in\mathcal{Q}}\tilde{X}(Q)\cdot L^{+}(Q)\nonumber\\
	& =&d_{G}(t,t')+\frac{4c_{d}}{\cint\delta}\cdot f\left(\{\tilde{X}(Q)\}_{Q\in\mathcal{Q}}\right)=O\left(\ln k\right)\cdot d_{G}(t,t')~.\nonumber
\end{eqnarray}
\begin{figure}[H]
	\centering{\includegraphics[scale=0.85]{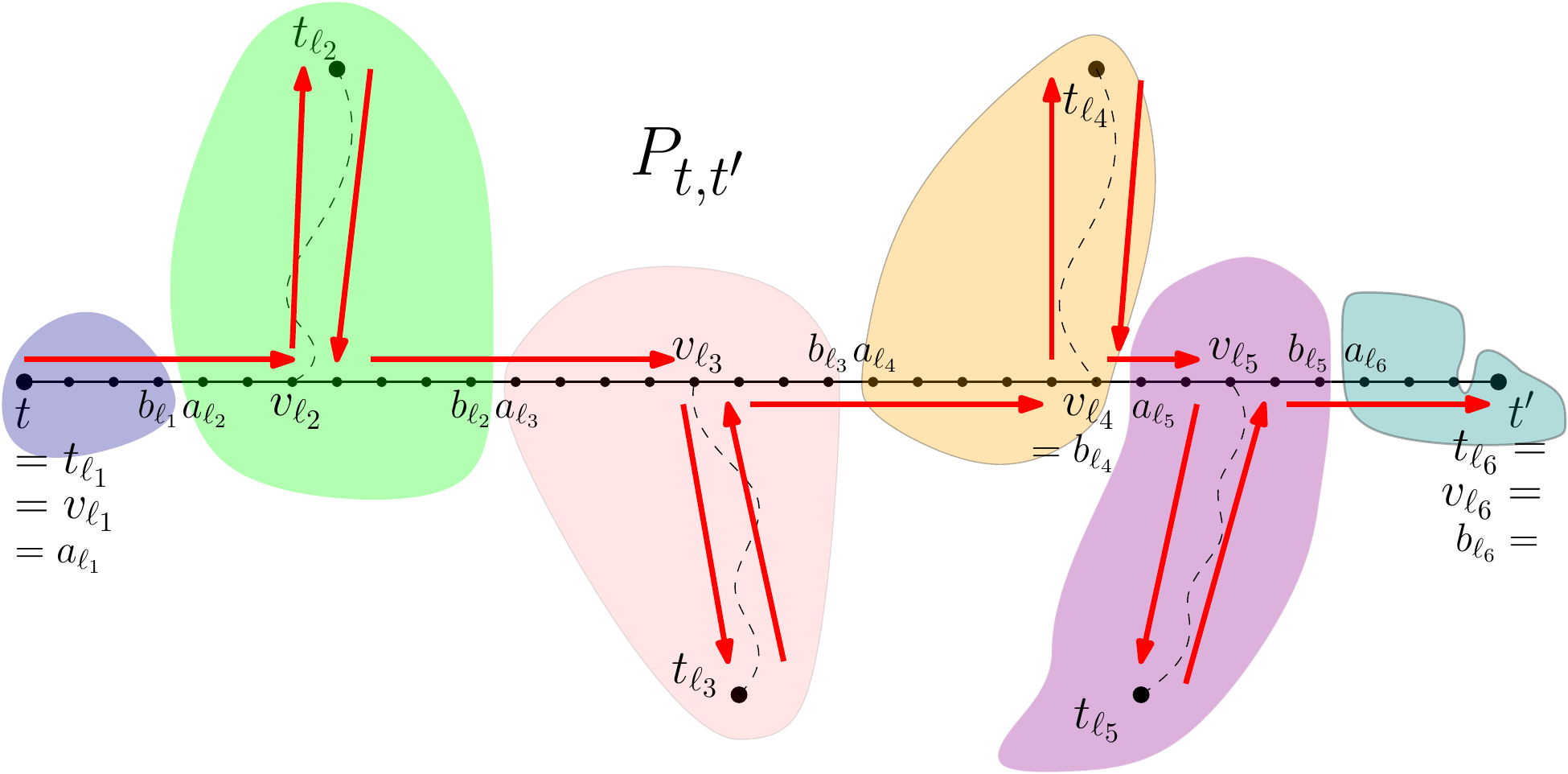}} 
	\caption{\label{fig:distortion}\small
		\it 
		The vertices $P_{t,t'}=v_{0}\dots v_{\gamma}$ are
		divided into consecutive detours $\mathcal{D}_{\ell_1},\dots,\mathcal{D}_{\ell_6}$. $t_{\ell_1},t_{\ell_2},t_{\ell_3},t_{\ell_4},t_{\ell_5},t_{\ell_6}$ is a path in the terminal-centered minor $M$ of $G$ (induced by $V_1,\dots,V_k$).
		The weight of the edge $\{t_{\ell_j},t_{\ell_{j+1}}\}$ in $M$ is $d_G(t_{\ell_j},t_{\ell_{j+1}})$, which is bounded by $d_{G}(t_{\ell_{j}},v_{\ell_{j}})+d_{G}(v_{\ell_{j}},v_{\ell_{j+1}})+d_{G}(v_{\ell_{j+1}},t_{\ell_{j+1}})$.
	}
\end{figure}
\end{proof}

As $\Pr\left[\overline{\EB}\wedge\overline{\EfBig}\right]\ge1-\left(\Pr\left[\EB\right]+\Pr\left[\EfBig\right]\right)\ge1-\frac{1}{2k}-\frac{1}{2k}=1-\frac{1}{k}$, \Cref{thm:mainSPR} follows.

\section{\texttt{Fast-Noisy-Voronoi} Algorithm}\label{sec:fastNV}
In this section, we describe a slightly modified version of the \texttt{Noisy-Voronoi} algorithm. Then we will show how to implement the modified algorithm in $O(m\log n)$ time.

Given two terminals $t_i,t_j$, and two clusters $V_i,V_j\subseteq V$ s.t. $t_i$ (resp $t_j$) is the unique terminal in $V_i$ (resp. $V_j$), $d_{G,V_i+V_j}(t_i,t_j)$ denotes the length of the shortest path between $t_i$ and $t_j$ in $G[V_i\cup V_j]$ that uses exactly one crossing edge between $V_i$ to $V_j$. See \Cref{fig:SingleCross} for an illustration.
\begin{figure}[h]
	\centering{\includegraphics[scale=0.65]{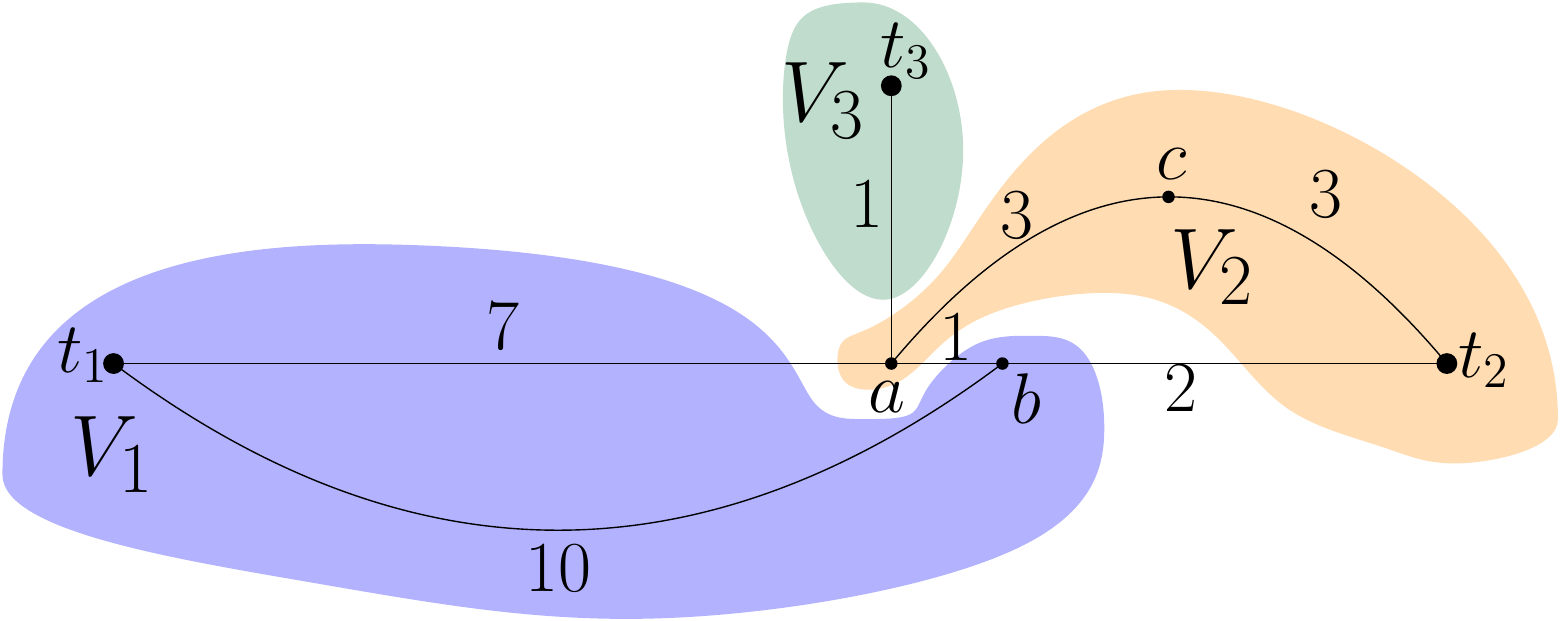}} 
	\caption{\label{fig:SingleCross}\small \it 
		$t_1,t_2,t_3$ are terminals. The different color areas describes the terminal partition. The shortest path in $G$ from $t_1$ to $t_2$ is $t_1,a,b,t_2$ and has length $d_G(t_1,t_2)=10$.
		Note that all the vertices in this path are in $V_1\cup V_2$.
		Nevertheless, the shortest path from $t_1$ to $t_2$ that uses only one crossing edge from $t_1$ to $t_2$ is $\{t_1,b,t_2\}$ and has length  $d_{G,V_1+V_2}(t_1,t_2)=12$.
	}
\end{figure}

In order to allow fast implementation, and avoid costly shortest path computations, we will introduce several modifications:
\begin{itemize}
	\item In \Cref{alg:mainSPR}, \lineref{line:returnMinor}, we will modify the edge weights in the induced terminal-centered minor. The weight of the edge $\{t_i,t_j\}$ (if exists)  will be $d_{G,V_i+V_j}(t_i,t_j)$ instead of $d_{G}(t_i,t_j)$. 
	\item In \Cref{alg:CreateCluster},  \lineref{AlgLine:CreCluPickvN}, instead of extracting an arbitrary vertex $v$ from $N$, we will extract the closest vertex $v$ to $t_j$ in $N$ w.r.t. the shortest path metric induced by $V_j\cup\{v\}$ (i.e. $v\in N$ with minimal $d_{G[V_j\cup \{v\}]}(v,t_j)$, note that it is a different graph for each vertex).\\
	Similarly, in \lineref{line:desidion}, instead of checking whether $d_{G}(v,t_j)\le R_j\cdot D(v)$, we will check whether $d_{G[V_j\cup\{v\}]}(v,t_j)\le R_j\cdot D(v)$.
\end{itemize}
The pseudo-code of the modified algorithm appears in \Cref{alg:fastSPR} and \Cref{alg:FastCreateCluster}.
\begin{algorithm}[!ht]
	\caption{$M=\texttt{Fast-Noisy-Voronoi}(G=(V,E,w),K=\{t_1,\dots,t_k\})$}\label{alg:fastSPR}
	\begin{algorithmic}[1]
		\STATE Set $\delta = \frac{1}{20\ln k}$ and $p=\frac15$. 
		\STATE Set $V_\perp~\la~V\setminus  K$.
		\hfill\emph{//~$V_\perp$ is the currently unclustered vertices.}					
		\FOR {$j$ from $1$ to $k$}
		\STATE Choose independently at random $g_j$ distributed according to $\Geo(p)$. 
		\STATE Set $R_j\la (1+\delta)^{g_j}$.	
		\STATE Set $V_j\la \texttt{Fast-Create-Cluster}(G,V_\perp ,t_j,R_j)$.
		\STATE Remove all the vertices in $V_j$ from $V_\perp$.
		\ENDFOR
		\STATE Let $M$ be the minor of $G$ created by contracting all the internal edges in $V_1,\ldots,V_k$. The weight of the edge $\{t_i,t_j\}$ (if exists) is defined to be $d_{G,V_i+V_j}(t_i,t_j)$.\label{line:fastComputeM}
		\RETURN $M$.
	\end{algorithmic}	
\end{algorithm}
\begin{algorithm}[!ht]
	\caption{$V_j=\texttt{Fast-Create-Cluster}(G=(V,E,w),V_\perp ,t_j,R_j)$}\label{alg:FastCreateCluster}
	\begin{algorithmic}[1]
		\STATE Set $V_j\leftarrow \{t_j\}$.
		\STATE Set $U\la \emptyset$.
		\hfill\emph{//~$U$ is the set of vertices already denied from $V_j$.}					
		\STATE Set $N$ to be all the neighbors of $t_j$ in $V_\perp$.
		\WHILE {$N\ne \emptyset$}
		\STATE Let $v\in N$ be the vertex with minimal $d_{G[V_j\cup \{v\}]}(v,t_j)$. \label{AlgLine:FastCreCluPickvN}
		\STATE Remove $v$ from $N$. \label{line:exstractMin}
		\IF {$d_{G[V_j\cup\{v\}]}(v,t_j)\le R_j\cdot D(v)$} \label{line:Fastdesidion}
		\STATE Add $v$ to $V_j$.
		\STATE Add all the neighbors of $v$ in $V_\perp\setminus U$ to $N$.  \label{line:FastAddNeighbors}
		\ELSE
		\STATE Add $v$ to $U$.
		\ENDIF
		\ENDWHILE
		\RETURN $V_j$.
	\end{algorithmic}	
\end{algorithm}

\begin{theorem}\label{thm:fastSPR}
	With probability $1-\frac{1}{k}$, for the minor graph $M$ returned by \Cref{alg:fastSPR}, it holds that for every two terminals $t,t'$, $d_{M}(t,t')\le O\left(\log k\right)\cdot d_{G}(t,t')$.
	Moreover, executing \Cref{alg:fastSPR} takes 
	$O(m+\min\left\{ m,nk\right\}\cdot \log n)$ time.
\end{theorem}
We prove \Cref{thm:fastSPR} in several steps. First, in \subsectionref{subsec:BasicProp} we show that \Cref{alg:fastSPR} indeed returns a terminal partition, and that similarity to \Cref{alg:mainSPR}, the edge subdivision does not change the outcome of the algorithm. 
Then  in \subsectionref{subsec:fastDistortion} we'll go through the analysis provided in \Cref{sec:distortion}, and verify that it is still goes through for  \Cref{alg:fastSPR} as well.
Finally, in \subsectionref{SubSec:runtime} we describe an efficient implementation of \Cref{alg:fastSPR}. 

\subsection{Basic Properties}\label{subsec:BasicProp}
Consider the \texttt{Fast-Create-Cluster} procedure (\Cref{alg:FastCreateCluster}). This is a Dijkstra-like algorithm. For every vertex $v$, set $\ell_v=d_{G[V_j\cup \{v\}]}(v,t_j)$. Note that for a vertex $v$, the value $\ell_v$ is decreasing throughout the algorithm as the set $V_j$ grows.
Note also that $\ell_v$ is defined for all the vertices (but simply has value $\infty$ for vertices out of $V_j\cup N$).
Denote by $\hat{\ell}_v$ the value $\ell_v$ at the time $v$ is extracted from $N$ at \lineref{line:exstractMin} of \Cref{alg:FastCreateCluster} (if such an occasion indeed occurs).
\begin{claim}\label{clm:BasicProp}
	Consider the values $\hat{\ell}_v$ of the vertices, extracted from $N$ at \lineref{line:exstractMin} of \Cref{alg:FastCreateCluster}. Then this values are non-decreasing.
	That is, if $v$ was extracted before $v'$, then  $\hat{\ell}_v\le \hat{\ell}_{v'}$.\\
	Moreover, after $v$ is extracted, the value $\ell_v$ remains unchanged till the end of the algorithm.
\end{claim}
\begin{proof}
	The proof of the first property is by induction on the execution of the algorithm. Let $v,v'$ be a pair of vertices such that $v'$ was extracted from $N$ right after $v$. It will be enough to show that $\hat{\ell}_v\le \hat{\ell}_{v'}$.
	Consider the time when $v$ was extracted from $N$. Let $\tilde{V}_j$ denote the set $V_j$ at that time. By minimality, for every $u\in N$, $\hat{\ell}_v=d_{G[\tilde{V}_j\cup \{v\}]}(v,t_j)\le d_{G[\tilde{V}_j\cup \{u\}]}(u,t_j)$. 
	If the value $\ell_{v'}$ did not change, we already have  $\hat{\ell}_{v'}=d_{G[\tilde{V}_j\cup \{v'\}]}(v',t_j)\ge \hat{\ell}_v$ (as necessarily $v'\in N$ because it is extracted next).
	Otherwise, if the value $\ell_{v'}$ decreased, then necessarily $v$ joined $V_j$ and the shortest path from from $t_j$ to $v'$ (in $\tilde{V}_j\cup\{v,v'\}$) goes through $v$ (as otherwise $\ell_{v'}$ would not have changed). 
	In particular, $\hat{\ell}_{v'}=d_{G[\tilde{V}_j\cup \{v,v'\}]}(t_j,v')=d_{G[\tilde{V}_j\cup \{v,v'\}]}(t_j,v)+d_{G[\tilde{V}_j\cup \{v,v'\}]}(v,v')> \hat{\ell}_v$.
	
	For the second property (that after extraction, $\ell_v$ remains unchanged), seeking contradiction, assume that $\ell_v$ is updated after some $u$ is extracted from $N$ and joined $V_j$. This implies that the new shortest path from $t_j$ to $v$ goes trough $u$, and thus is of length greater than $\hat{\ell}_u$, a contradiction.
\end{proof}

Now we are ready to show that \Cref{alg:fastSPR} indeed returns a terminal partition (that is, reprove \Cref{lem:AlgRetPar}).
\begin{lemma}\label{lem:FastAlgRetPar}
	The sets $V_1,\ldots,V_k$ constructed by \Cref{alg:fastSPR} constitutes a terminal partition.
\end{lemma}
\begin{proof}
	It is clear that the clusters $V_1,\dots,V_j$ are disjoint, and that each cluster is connected. It will be enough to argue that every vertex $v\in V$ is clustered.
	Following along the lines of the proof of \Cref{lem:AlgRetPar}, let $t_j$ be the closest terminal to $v$, and $P=\{t_j=u_0,u_1,\dots ,u_s=v\}$ be the shortest path from $t_j$ to $v$. 
	Let $u_{i'}$ be the first vertex from $P_{t,t'}$ to be clustered during the algorithm ($u_0=t_j\in V_j$, so at least one vertex in $P_{t,t'}$ is clustered). Let $V_{j'}$ be the cluster $u_{i'}$ joins to.  We argue by induction on $i\ge i'$ that $u_{i}$ also joins $V_{j'}$. This will imply that $u_s=v$ joins $V_{j'}$ and thus is clustered.
	
	Suppose $u_i$ joins $V_{j'}$. Denote by $V_{j'}^i$ the set $V_{j'}$ right after $u_i$ joins it. 
	As $u_i$ joins $V_{j'}$, $d_{G[V_{j'}^i]}(u_i,t_{j'})\le R_{j'}\cdot D(u_i)$.
	In particular, at that stage 
	\begin{align*}
	\ell_{u_{i+1}}= d_{G\left[V_{j'}^{i}\cup\left\{ u_{i+1}\right\} \right]}(u_{i+1},t_{j'}) & \le d_{G\left[V_{j'}^{i}\right]}(u_{i},t_{j'})+w\left(\left\{ u_{i},u_{i+1}\right\} \right)\\
	& \le R_{j'}\cdot D(u_{i})+d_{G}(u_{i},u_{i+1})\le R_{j'}\cdot D(u_{i+1})~,
	\end{align*}
	As at least one neighbor ($u_{i}$) of $u_{i+1}$ joins $V_{j'}$, $u_{i+1}$ joins $N$ at some stage of the algorithm. In particular, by \Cref{clm:BasicProp}, when $u_{i+1}$ will be extracted from $N$, $\hat{\ell}_{u_{i+1}}\le  R_{j'}\cdot D(u_{i+1})$, and thus $u_{i+1}$ will join $V_{j'}$ as required. 
\end{proof}

We will use the modified graph $\hat{G}$ (with the subdivided edges) for the distortion analysis. In order to prove validity, we will argue that \Cref{clm:SameMinor} still holds.
\begin{claim}\label{clm:FastSameMinor}
	In \Cref{clm:SameMinor}, if we replace \Cref{alg:mainSPR} with \Cref{alg:fastSPR}, the claim still holds.
\end{claim}
\begin{proof}
	We follow the lines of the proof of \Cref{clm:SameMinor}.
	Let $V_1,\dots,V_k$ (resp. $\tilde{V}_1,\dots,\tilde{V}_k$) be the terminal partition induced by \Cref{alg:fastSPR} on $G$ (resp. $\tilde{G}$).
	We argue that for all $j$, $ V_j=\tilde{V}_j\setminus\{v_e\}$.
	As previously, this will imply that the terminal-centered minors have the same edges set. As $v_e$ only subdivides the edge $e$, it will also hold for all $i,j$ that $d_{G,V_i+V_j}(t_i,t_j)=d_{G,\tilde{V}_i+\tilde{V}_j}(t_i,t_j)$, and thus the edge weights in both minors will also be identical. In particular, the claim will follow.
	
	Suppose w.l.o.g that $v$ joins $V_j$ while $u$ is still unclustered. Denote by $V_j'$ (resp. $\tilde{V}_j'$) the set $V_j$ (resp. $\tilde{V}_j$) right after the clustering of $v$ at the execution of \Cref{alg:fastSPR} on $G$ (resp. $\tilde{G}$).
	As previously,  for all $j''<j$, $V_{j''}=\tilde{V}_{j''}$, while $V_j'=\tilde{V}_j'$.
	
	Recall that $\hat{\ell}_v=d_{G[V'_j](t_j,v)}$ (resp. 	$\tilde{\hat{\ell}}_v$) denotes the distance between $t_j$ to $v$ at the time of the extraction of $v$ from $N$ (resp. $\tilde{N}$). Note that $\hat{\ell}_v=\tilde{\hat{\ell}}_v$.
	As $v$ joins $V_j$, necessarily $\hat{\ell}_v\le R_j\cdot D(v)$.
	In the rest of the proof we consider the following cases:
	\begin{itemize}
		\item \textbf{$\hat{\ell}_u>R_j\cdot D(v)$ :} In this case $u$ will not join $V_j$.
		As $v_e$ has edges only to $v$ and $u$, $v_e$ has no impact on any other vertex. In particular, $\hat{\ell}_u\le\tilde{\hat{\ell}}_u$. Therefore 		
		$\tilde{V}_j$ will be constructed in the same manner as $V_j$ (up to maybe containing $v_e$).
		Note that all the other clusters will not be effected, as if $v_e$ remained unclustered, it becomes a leaf. We conclude that for every $j'$, $V_{j'}=\tilde{V}_{j'}\setminus\{v_u\}$. 		
		
		\item \textbf{$\hat{\ell}_u\le R_j\cdot D(v)$ :} Recall that $\omega$ is the weight of $e$. There are two sub-cases:
		\begin{itemize}
			\item {$\hat{\ell}_u= \hat{\ell}_v+\omega$ -}
			After $v$ joins $\tilde{V}_j$, the label of $v_e$ is updated to $\hat{\ell}_{v_e}\leftarrow \tilde{\hat{\ell}}_v+\frac{\omega}{2}$.	It holds that 	
			\begin{align*}
			\tilde{\hat{\ell}}_{v_e} & \le\tilde{\ell}_{v_e}=\tilde{\hat{\ell}}_{v}+\frac{\omega}{2}=\hat{\ell}_{v}+\frac{\omega}{2}=\frac{1}{2}\left(\hat{\ell}_{v}+\hat{\ell}_{u}\right)\\
			& \le\frac{1}{2}\cdot R_{j}\left(D(v)+D(u)\right)\le R_{j}\cdot D(e_{v})~.
			\end{align*}		
			In particular, $v_e$ will join $\tilde{V}_j$, and $\tilde{\ell}_u$ will be updated to $\tilde{\hat{\ell}}_{v_e}+\frac\omega2=\tilde{\hat{\ell}}_v+\omega$. From this point on, the two algorithms will behave in the same way. In particular, for every $j''\ne j$, $V_{j''}=\tilde{V}_{j''}$ while $V_{j}\cup\{v_e\}=\tilde{V}_{j}$.
			\item {$\hat{\ell}_u< \hat{\ell}_v+\omega$ -} It holds that $u$ joins $V_j$. However, the shortest path in $V_j$ from $t_j$ to $u$ did not goes trough $v$. Therefore, as $v_e$ did not effect any vertex (other than $v,u$), the execution will proceed in the same way in both algorithms, and $u$ will join $\tilde{V}_j$.
			As each cluster is connected and all the vertices are clustered, necessarily $v_e$ will  join $\tilde{V}_j$ as well.
			We conclude that for every $j''\ne j$, $V_{j''}=\tilde{V}_{j''}$ while $V_{j}\cup\{v_e\}=\tilde{V}_{j}$.  
		\end{itemize}
	\end{itemize} 	 
\end{proof}

\subsection{Distortion Analysis}\label{subsec:fastDistortion}
We will follow the distortion analysis of \Cref{alg:mainSPR} given in \Cref{sec:distortion}. Consider two terminals $t,t'$. We will use the exact same notation (the reader is suggested to refer to \Cref{appendix:key} in order to recall notations and definitions).
We start by reproving \Cref{clm:failProb}.
\begin{claim}\label{clm:fastfailProb}
	During the execution of \Cref{alg:fastSPR}, assuming $R_j\ge r_{v^j}$, all of $S_j$ joins $V_j$ with probability at least $1-p$.
\end{claim}
\begin{proof}
	Denote $S_j=\{u_{j-q'},\dots,u_{j},\dots,u_{j+q}\}\subseteq Q_j\subseteq P_{t,t'}$ where $v^j=u_j$.
	Denote by $V_j'$ the cluster $V_j$ right after $u_j$ joins.
	As $u_j$ joined, necessarily $\frac{d_{G[V_{j}'\cup\{u_j\}]}(u_{j},t_{j})}{D(u_{j})}\le r_{v^{j}}\le R_{j}$. 
	We will denote by $\bar{V}_j$ the cluster $V_j$ at the end of the algorithm.
	Following inequality \eqref{eq:RjBigProb}, with probability $1-p$, $R_j\ge(1+\delta)r_{v^j}$.
	We will show that if this event indeed occur, then $S_j\subseteq \bar{V}_j$.
	
	We argue by induction on $i$, that $u_{j+i}\in \bar{V}_j$. The proof that $u_{j-i}\in \bar{V}_j$ is symmetric. 
	Assume that  $\{u_i,u_{i+1},\dots,u_{j+i-1}\}\subseteq \bar{V}_{j}$. 
	Following inequalities \eqref{eq:BoundLQj} and \eqref{eq:BoundDu} , $L(Q_{j})\le 2\cint\delta\cdot D(v^{j})$ and $D(u_{j+i})\ge D(v^{j})\left(1-2\cint\delta\right)$.
	As $u_{i+j-1}\in\bar{V}_{j}$, $u_{j+i}$ necessarily joins $N$ at some stage. In particular, at the time $u_{j+i}$ was extracted from $N$, 
	\[
	\hat{\ell}_{u_{j+i}}= d_{G\left[\bar{V}_{j}\cup\left\{ u_{j+i}\right\} \right]}(t_{j},u_{j+i})\le d_{G\left[V_{j}'\right]}(t_{j},v^{j})+L(Q_{j})\le d_{G\left[V_{j}'\right]}(t_{j},v^{j})\left(1+2\cint\delta\right)~,
	\]
	where the first equality follows by \Cref{clm:BasicProp}, as $\hat{\ell}_{u_{j+i}}$ remains unchanged after extraction.
	We conclude that
	\[
	\frac{\hat{\ell}_{u_{j+i}}}{D(u_{j+i})}\le\frac{d_{G\left[V_{j}'\right]}(t_{j},v^{j})\left(1+2\cint\delta\right)}{D(v^{j})\left(1-2\cint\delta\right)}\le\frac{d_{G\left[V_{j}'\right]}(t_{j},v^{j})}{D(v^{j})}\left(1+3\cdot2\cint\delta\right)\le\left(1+\delta\right)R_{j}~.
	\]
	We conclude that $u_{j+i}$ joins $V_j$ as required.
\end{proof}

In \subsectionref{subsec:FailCount} we defined charge function $f(\{x_Q\}_{Q\in\mathcal{Q}})=\sum_{Q\in\mathcal{Q}}X(Q)\cdot L^{+}(Q)$, and in \Cref{lem:fbound} we upper bounded its value (w.h.p). In that analysis we exploit only \Cref{clm:failProb}. Replacing it with \Cref{clm:fastfailProb}, the analysis still hold. That is 
	$\Pr\left[f\left(\{\tilde{X}(Q)\}_{Q\in\mathcal{Q}}\right)\ge43\cdot d_G(t,t') \right]\le k^{-3}$.
Denote by $\EfBig$  the event that for some pair of terminals $t,t'$, 
$f\left(\tilde{X}(Q^{1}),\dots,\tilde{X}(Q^{\varphi})\right)\ge43 \cdot d_G(t,t')$ . As previously, by union bound $\Pr\left[\EfBig\right]< \frac{1}{2k}$.
Denote by $\EB$  the event that for some $j$, $R_j> c_d$.
By \Cref{clm:boundEB}, $\Pr[\EB]\le \frac{1}{2k}$.
We argue that assuming $\overline{\EB}$ and $\overline{\EfBig}$ (which happens with probability $1-\frac1k$), the distance between every pair of terminals $t,t'$ in the minor returned by \Cref{alg:fastSPR} bounded by $O(\log k)\cdot d_G(v,u)$.
This will conclude the proof of the distortion argument in \Cref{thm:fastSPR}.
Recall that in contrast to \Cref{alg:mainSPR}, the weight of the edge $\{t_i,t_j\}$ (if exists) 
is $d_{G,V_i+V_j}(t_i,t_j)$ rather than $d_{G}(t_i,t_j)$, this will force some changes to our analysis.
Recall the notations we used in \Cref{lem:distortion}: the path $P_{t,t'}$ is divided into consecutive detours $\mathcal{D}_{\ell_1},\dots,\mathcal{D}_{\ell_{k'}}$.
The leftmost (resp. rightmost) vertex in $\mathcal{D}_{\ell_j}$ denoted by $a_{\ell_j}$ (resp. $b_{\ell_j}$). Both $a_{\ell_j},b_{\ell_j}$ belong to $V_{\ell_j}$, the cluster of $t_{\ell_j}$. 
In particular, the graph $G$ contains an edge between $b_{\ell_j}$ to $a_{\ell_{j+1}}$. Recall also that $t_{\ell_1}=t$ and $t_{\ell_k'}=t'$ (as each terminal covers itself). 
It holds that,
\begin{align*}
d_{M}(t,t') & \le\sum_{j=1}^{k'-1}d_{G,V_{\ell_{j}}+ V_{\ell_{j+1}}}(t_{\ell_{j}},t_{\ell_{j+1}})\\
& \le\sum_{j=1}^{k'-1}\left[d_{G\left[V_{\ell_{j}}\right]}(t_{\ell_{j}},b_{\ell_{j}})+d_{G}(b_{\ell_{j}},a_{\ell_{j+1}})+d_{G\left[V_{\ell_{j+1}}\right]}(a_{\ell_{j+1}},t_{\ell_{j+1}})\right]\\
& \le c_{d}\cdot\sum_{j=1}^{k'-1}\left[d_{G}(t_{\ell_{j}},b_{\ell_{j}})+d_{G}(b_{\ell_{j}},a_{\ell_{j+1}})+d_{G}(a_{\ell_{j+1}},t_{\ell_{j+1}})\right]\\
& \le c_{d}\cdot\sum_{j=1}^{k'-1}\left[d_{G}(t_{\ell_{j}},v^{\ell_{j}})+d_{G}(v^{\ell_{j}},b_{\ell_{j}})+d_{G}(b_{\ell_{j}},a_{\ell_{j+1}})+d_{G}(a_{\ell_{j+1}},v^{\ell_{j+1}})+d_{G}(v^{\ell_{j+1}},t_{\ell_{j+1}})\right]\\
& \le c_{d}\cdot\left(\sum_{j=1}^{k'-1}d_{G}(v^{\ell_{j}},v^{\ell_{j+1}})+2\sum_{j=1}^{k'}d_{G}(t_{\ell_{j}},v^{\ell_{j}})\right)\\
& \le c_{d}\cdot\left(d_{G}(t,t')+2c_{d}\cdot\sum_{j=1}^{k'}D(v^{\ell_{j}})\right)\\
& =O\left(\ln k\right)\cdot d_{G}(t,t')~.
\end{align*}
The third inequality follows by our assumption $\overline{\EB}$, as for every index $j$ and vertex $v\in V_j$, it holds that
$d_{G\left[V_{j}\right]}(t_{j},v)\le c_{d}\cdot D(v)\le c_{d}\cdot d_{G}(t_{j},v)$. The fifth inequality follows as all $v^{\ell_{j}},b_{\ell_{j}},a_{\ell_{j+1}},v^{\ell_{j+1}}$ lie on the same shortest path $P_{t,t'}$.
The sixth inequality follows by $\overline{\EB}$ as $d_{G}(t_{\ell_{j}},v^{\ell_{j}})\le d_{G\left[V_{\ell_{j}}\right]}(t_{\ell_{j}},v^{\ell_{j}})\le c_{d}\cdot D(v^{\ell_{j}})$.
The equality follows by inequality (\ref{eq:distBound2}) and $\overline{\EfBig}$.

\subsection{Runtime}\label{SubSec:runtime}
For the implementation of \Cref{alg:fastSPR} and the \texttt{Fast-Create-Cluster} procedure we will use two basic data structures. The first one is a binary array to determine set membership of the vertices. It is folklore (see for example \cite{AH74}) that an array could be initialized in constant time to be the all $0$ array (that is the empty set). Changing entry (that is adding or deleting an element) also takes constant time.
The second data structure is the Fibonacci heap (see \cite{FT87}). Here each  element has a key (some real number), and we can add new element or decrease the value of the key in constant time.  Finding the minimal element in the heap and deleting it takes $O(\log h)$ time (assuming there are currently $h$ elements in the heap).

Before the execution of \Cref{alg:fastSPR}, we compute the values $D(v)$ for all $v\in V$. This is done using an auxiliary graph $G'$ where we add new vertex $s$ with edges of weight $0$ to all the terminals. Note that for every vertex $v$, the distance from $s$  exactly equals $D(v)$. Thus we can simply run Dijkstra algorithm from $s$ to determine $D(v)$ for all $v\in V$.  The runtime is $O(m+n\log n)$ (see \cite{FT87}).

Next we give a detailed implementation of the \texttt{Fast-Create-Cluster} procedure.
The sets $V_j,U$ and $V_\perp$ are stored using the arrays described above ($V_\perp$ will be a global variable). The set $N$ will be stored using Fibonacci heap, where the key value of $v\in N$ will be $\ell_v$ (i.e.  $d_{G[V_j\cup\{v\}]}(v,t_j)$). 
Denote by $\mathcal{N}_j$ all the elements who belong to $N$ at any stage of the execution of the \texttt{Fast-Create-Cluster} procedure (which created $V_j$). Let $m_j$ denote the number of edges incident on vertices of $V_j$. 
Each iteration of the while loop starts by deleting an element $v$ with minimal key (of value $\hat{\ell}_v$) from $N$ ($O(\log |\mathcal{N}_j|)$ time). Then we examine whether to add $v$ to $V_j$ (in $O(1)$ time). If $v$ is rejected, we add $v$ to $U$  (in $O(1)$ time). Otherwise, $v$ is added to $V_j$. In the latter case we go over each neighbor $u$ of $v$. 
If $u\in U$ we do nothing. If $u\in N$, its key $\ell_u$ is updated to be $\min\{\ell_u,\ell_v+w(\{v,u\})\}$. Finally, if $u\in V_\perp\setminus (U\cup N)$, then $u$ is added to $N$ with the key $\ell_u\leftarrow\ell_v+w(\{v,u\})$.
It is easy to verify that all the keys are indeed maintained with the correct values. 
Note that all this processing for $u$ takes only $O(1)$ time. 
In particular, processing all neighbors throughout the \texttt{Fast-Create-Cluster} procedure takes $O(m_j)$ time. All the deletion of elements from the heap $N$ takes $O(|\mathcal{N}_j|\log|\mathcal{N}_j|)$ time. 

Next we bound the total cost of the $k$ calls to the \texttt{Fast-Create-Cluster} procedure. $|\mathcal{N}_j|$ can be bounded from above by both $m_j$ and $n$. 
Moreover, $\sum_jm_j\le 2m$, as every edge is incident on only two vertices.
We provide two upper bounds on the running time:
\begin{align*}
O(n)+\sum_{j=1}^{k}\ensuremath{O(m_{j}+|\mathcal{N}_{j}|\log|\mathcal{N}_{j}|)} & \le O\left(m+\sum_{j=1}^{k}m_{j}\log n\right)=O(m\log n)~.\\
O(n)+\sum_{j=1}^{k}\ensuremath{O(m_{j}+|\mathcal{N}_{j}|\log|\mathcal{N}_{j}|)} & \le O\left(m+\sum_{j=1}^{k}n\log n\right)=O(m+nk\log n)~.
\end{align*}
Thus the total running time of this $k$ calls bounded by $O(m+\min\left\{ m,nk\right\} \cdot\log n)$.
Finally we bound the total runtime of \Cref{alg:fastSPR} without the calls to the  \texttt{Create-Cluster}.
It is straightforward that up to \lineref{line:fastComputeM}, where we create the minor $M$ given the clusters, all computations took $O(n)$ time\footnote{In fact, the sampling of $g_1,\dots,g_k$ takes $O(k)$ time only with high probability. But we will ignore this issue.}.
Using \Cref{clm:BasicProp}, by the end of the for loop in \Cref{alg:fastSPR}, for every $j$ and $v\in V_j$ it holds that $\hat{\ell}_v=d_{G[V_j]}(t_j,v)$.
In order to create the minor graph $M$, we go over all the edges iteratively, for every edge $\{v,u\}\in E$, such that $v\in V_j$, $u\in V_i$ and $i\ne j$. We add an edge $\{t_i,t_j\}$ to $M$ (if it does not exist already). The weight of the edge updated to be the minimum between the current weight ($\infty$ if it does not exist yet) and $\hat{\ell}_v+w(\{v,u\}) +\hat{\ell}_u$ (the keys at the time of extraction from $N$).
It is straightforward that by the end of this procedure we will indeed compute the minor $M$, and each edge  $\{t_i,t_j\}$ in $M$ will have weight $d_{G,V_i+ V_j}(t_i,t_j)$.
This iterative process takes $O(m)$ time. \Cref{thm:fastSPR} now follows.

\section{Lower bounds on the Performance of the Algorithms}\label{sec:LB} 
Chan et. al. \cite{CXKR06} gave a lower bound of $8$ for the distortion in the Steiner Point Removal problem. This lower bound was not improved since.
This section is dedicated to lower bound the performance of the various algorithms which were suggested for the problem. That is, while we do not provide better lower bounds for the Steiner Point Removal problem itself, we are able to lower bound the performance of the algorithms used so far.

In \subsectionref{subsec:NVLB} we prove that our analysis of the \texttt{Noisy-Voronoi} algorithm (\Cref{alg:mainSPR}\&\Cref{alg:fastSPR}) is asymptotically tight. That is, there is a graph family on which the achieved distortion is $\Theta(\log k)$. 
Next, in \subsectionref{subsec:BGLB}, we provide a lower bound on the performance of the \texttt{Ball-growing} algorithm studied by \cite{KKN15,Che18,Fil18}. Specifically, we provide (the same) graph family on which the \texttt{Ball-growing} algorithm incurs $\Omega(\sqrt{\log k})$ distortion.
Recall that in \cite{Fil18}, the author proved that the \texttt{Ball-growing} algorithm finds a minor with distortion $O(\log k)$. That is, while the analysis of the \texttt{Ball-growing} algorithm still might be improved, it cannot be pushed further than $\Omega(\sqrt{\log k})$.

First, we show that the \emph{expected} distortion incurred by the minor returned by the algorithms is large. Then, we deduce that with constant probability the (usual-worst case) distortion is also large.
Formally, both the algorithms are randomized, and thus can be viewed as producing a distribution $\mathcal{D}$ over graph minors. Given such distribution $\mathcal{D}$, the expected distortion of the pair $t,t'$ is $\mathbb{E}_{M\sim\mathcal{D}}\left[\frac{d_M(t,t')}{d_G(t,t')}\right]$. The overall expected distortion is the maximal expected distortion among all terminal pairs.

A final remark: both algorithms used an arbitrary order over the terminals, in contrast to similar algorithms for other problems \cite{CKR01,FRT04} which consider a random order.
Our lower-bounds will still hold even if one replaces the arbitrary order with a random one.

\subsection{Lower bound on the performance of the \texttt{Noisy-Voronoi} algorithm}\label{subsec:NVLB}
The following theorem provides a lower bound on the expected distortion incurred by \Cref{alg:mainSPR}.
The graphs which we will use for the lower bound are trees. As both \Cref{alg:mainSPR} and \Cref{alg:fastSPR} are identical where the input graph is a tree, the lower bound will also hold on \Cref{alg:fastSPR}.
\begin{theorem}\label{thm:VorLB}
	Fix some $k\in\mathbb{N}$. There is a graph $G=(V,E,w)$ with terminal set $K$ of size $k$, such that the expected distortion of the minor returned by \Cref{alg:mainSPR} is $\Omega(\log k)$.
\end{theorem}
\begin{proof}	
	We will assume that $k$ is large enough, as otherwise $1=\Omega(\log k)$ and hence every graph with $k$ terminals provides a valid lower bound. 
	Let $G_k$ be the graph described in 
	\Cref{fig:VoronoiFail} with parameter $\eps=14\delta=\Theta(\frac{1}{\log k})$.
	Let $X_{j}$ be an indicator for
	the event $v_{j}\in V_j$, that is $t_j$ covers $v_{j}$. For $X_j$ to occur, it is enough that for every $i\ne j$, $d_G(t_i,v_j)>R_i\cdot D(v_j)$. That is $R_i<1+|i-j|\cdot\eps$. 
	By the definition of $R_i$, 
	\[
	\Pr\left[R_{i}\ge1+\left|i-j\right|\epsilon\right]=\Pr\left[g_{i}\ge\log_{1+\delta}\left(1+\left|i-j\right|\epsilon\right)\right]=\left(1-p\right)^{\left\lceil \log_{1+\delta}\left(1+\left|i-j\right|\epsilon\right)-1\right\rceil }~.
	\]	
	For $i$ such that $|i-j|<\frac1\eps$, it holds that $\log_{1+\delta}\left(1+\left|i-j\right|\epsilon\right)=\frac{\ln\left(1+\left|i-j\right|\epsilon\right)}{\ln\left(1+\delta\right)}\ge\frac{\left|i-j\right|\epsilon/2}{\delta}$
	. While for $i$ such that $|i-j|\ge\frac1\eps$, $\log_{1+\delta}\left(1+\left|i-j\right|\epsilon\right)\ge\frac{\ln2}{\ln1+\delta}\ge\frac{1}{2\delta}$.
	We conclude 	
	\begin{align*}
	\Pr\left[X_{i}\right] & \ge\Pr\left[\forall_{j\ne i}\left(R_{j}<1+\left|i-j\right|\epsilon\right)\right]\\
	& \ge1-\sum_{j\ne i}\Pr\left[R_{j}\ge1+\left|i-j\right|\epsilon\right]\\
	& \ge1-2\sum_{i=1}^{\left\lfloor \frac{1}{\epsilon}\right\rfloor }\left(\left(1-p\right)^{\frac{i\epsilon/2}{\delta}-1}\right)-k\left(1-p\right)^{\frac{1}{2\delta}-1}~.
	\end{align*}
	Now, $\sum_{i=1}^{\left\lfloor \frac{1}{\epsilon}\right\rfloor }\left(1-p\right)^{\frac{i\epsilon/2}{\delta}}\le\sum_{i=1}^{\infty}\left(\left(1-p\right)^{7}\right)^{i}\le\sum_{i=1}^{\infty}\frac{1}{4^{i}}=\frac{1}{4}\frac{1}{1-\frac{1}{4}}=\frac{1}{3}$. While $k\left(1-p\right)^{\frac{1}{2\delta}}=k\left(\frac{4}{5}\right)^{10\ln k}=k^{1-10\ln\frac{5}{4}}\le\frac{1}{k}$. In particular 
	$\Pr\left[X_{i}\right]\ge1-(1-p)^{-1}\cdot\left(2\cdot\frac{1}{3}+\frac{1}{k}\right)=\Omega(1)$.
	
	Set $X=\sum_{i=2}^{k-1}X_i$. By linearity of expectation, $\mathbb{E}[X]=\Omega(k)$. Note that the distance from $t_{1}$ to
	$t_{k}$ in the minor graph $M_k$ equals $2+\left(k-1\right)\epsilon+2X$. We
	conclude
	\[	\mathbb{E}\left[\frac{d_{M_{k}}(t_{1},t_{m})}{d_{G_{k}}(t_{1},t_{m})}\right]=\frac{2+\left(k-1\right)\epsilon+2\mathbb{E}\left[X\right]}{2+\left(k-1\right)\epsilon}=\frac{\Omega(k)}{O(k\epsilon)}=\Omega\left(\frac{1}{\epsilon}\right)=\Omega(\log k)~.\qedhere
	\]
\end{proof}
\begin{corollary}\label{cor:VorLB}
	Fix some $k\in\mathbb{N}$. There is a graph $G=(V,E,w)$ with terminal set $K$ of size $k$, such that with constant probability, the distortion incurred by the minor returned by \Cref{alg:mainSPR} is $\Omega(\log k)$.
\end{corollary}
\begin{proof}
	We will use the graph and notations from the proof of \Cref{thm:VorLB}. Set $\mu=\mathbb{E}\left[\frac{d_{M_{k}}(t_{1},t_{m})}{d_{G_{k}}(t_{1},t_{m})}\right]=\Omega(\log k)$.
	Note the largest possible distortion 
	is $\frac{2k-2+\left(k-1\right)\epsilon}{2+\left(k-1\right)\epsilon}=c\cdot\mu$,
	for some constant $c\ge1$ (this distortion occurred exactly when each
	vertex $v_{j}$ belongs to $V_{j}$). Denote by $\chi$ the event
	that $\frac{d_{M_{k}}(t_{1},t_{m})}{d_{G_{k}}(t_{1},t_{m})}\ge\frac{1}{2}\mu$.
	Then
	\[
	\mu=\mathbb{E}\left[\frac{d_{M_{k}}(t_{1},t_{m})}{d_{G_{k}}(t_{1},t_{m})}\right]\le\Pr\left[\chi\right]\cdot c\mu+\left(1-\Pr\left[\chi\right]\right)\cdot\frac{1}{2}\mu\,,
	\]
	therefore
	\[
	\Pr\left[\chi\right]\ge\frac{1-\frac{1}{2}}{c-\frac{1}{2}}\ge\frac{1}{2c}=\Omega(1)\,.
	\]
	Therefore, with constaint probability, the distortion is at least
	$\frac{1}{2}\mu=\Omega(\log k)$.
\end{proof}

\subsection{Lower Bound on the Performance of the \texttt{Ball-Growing} Algorithm}\label{subsec:BGLB}
In this subsection we provide a lower bound on the performance of the  \texttt{Ball-Growing} algorithm. 
For completeness, we attach in \Cref{sec:BallGrowing} a full description of the \texttt{Ball-Growing} algorithm as it appeared in
\cite{Fil18}. In particular, we will use the notations defined there.
The \texttt{Ball-Growing} as described in \cite{Fil18} also had a modification step. As our lower bound example is a tree, this modification has no impact on the minor returned by the algorithm, and thus we can ignore it.
Formally, a claim similar to \Cref{clm:SameMinor} can be proven.

\begin{theorem}\label{thm:BGLB}
	Fix some $k\in\mathbb{N}$. There is a graph $G=(V,E,w)$ with terminal set $K$ of size $k$, such that the expected distortion of the minor returned by the \texttt{Ball-Growing} algorithm is $\Omega(\sqrt{\log k})$.
\end{theorem}
\begin{proof}
	We will use the graph described in \Cref{fig:VoronoiFail} with modified parameters: the weight of an edge between terminal to Steiner vertex will be $2-\eps$ while the weight of an edge between two Steiner vertices will be $2\eps$ for $\eps$ to be specified later. Note that the \texttt{Ball-Growing} algorithm assumes that the minimal distance between a terminal to a Steiner vertex in the input graph is exactly $1$. In order to satisfy this condition we will add additional Steiner vertex as a leaf connected to $t_1$ via an edge of unit weight. Note that this new vertex has no impact on the resulting minor whatsoever, and therefore can be completely ignored.
	
	As previously, we denote by $X_j$ the indicator for the event $v_j\in V_j$.
	Following the analysis of \Cref{thm:BGLB}, if we will prove that $\Pr[X_j]=\Omega(1)$ (for arbitrary $j$) it will imply expected distortion of $\Omega(\frac1\eps)$.
	
	Let $\mathcal{R}_{j}$ be equal to $R_{j}$ (the magnitude of $t_{j}$)
	at the end of the $m=\log_{r}3-1$ round. For simplicity we will assume that $m$ is an integer, otherwise the analysis will go trough after slight modification of the parameters. Recall that $\mathcal{R}_{j}=\sum_{\ell=0}^{m}q_j^{\ell}$ where $q_j^\ell$ distributed
	according to $\text{Exp}(D\cdot r^{\ell})$. Here $r=1+\frac{\delta}{\ln k}$, $\delta=\frac{1}{20}$, $D=\frac{\delta}{\ln k}$, and all the $q_j^\ell$ are independent. 
	It holds that
	\begin{align*}
	\mathbb{E}\left[\mathcal{R}_{j}\right] & =\sum_{\ell=0}^{m}D\cdot r^{\ell}=D\cdot\frac{r^{m+1}-1}{r-1}=2\,.\\
	\mathbb{V}\left[\mathcal{R}_{j}\right] & =\mathbb{V}\left[\sum_{\ell=0}^{m}q_{j}^{\ell}\right]=\sum_{\ell=0}^{m}\mathbb{V}\left[q_{j}^{\ell}\right]=\sum_{\ell=0}^{m}\left(D\cdot r^{\ell}\right)^{2}\\
	& =D^{2}\cdot\frac{r^{2(m+1)}-1}{r^{2}-1}=\left(\frac{\delta}{\ln k}\right)^{2}\cdot\frac{9-1}{2\cdot\frac{\delta}{\ln k}+\left(\frac{\delta}{\ln k}\right)^{2}}\le4\cdot\frac{\delta}{\ln k}=O\left(\frac{1}{\ln k}\right)~.
	\end{align*}
	Where we used linearity of expectation and independence.
	In order that $X_j$ will occur, it is enough that $\mathcal{R}_{j}\ge d(t_{j},v_{j})$, while for every $j'\ne j$, $\mathcal{R}_{j}< d(t_{j'},v_{j})$.
	Using Chebyshev inequality, 
	\begin{align*}
	\Pr\left[\mathcal{R}_{j}\ge d(t_{j},v_{j})\right] & =\Pr\left[\mathcal{R}_{j}\ge2-\epsilon\right]\ge\Pr\left[\left|\mathcal{R}_{j}-\mathbb{E}\left[\mathcal{R}_{j}\right]\right|<\epsilon\right]\ge1-\frac{\mathbb{V}\left[\mathcal{R}\right]}{\epsilon^{2}}~.\\
	\Pr\left[\mathcal{R}_{j'}\ge d(t_{j'},v_{j})\right] & \le\Pr\left[\left|\mathcal{R}_{j'}-\mathbb{E}\left[\mathcal{R}_{j'}\right]\right|\ge\left(2\left|j-j'\right|-1\right)\epsilon\right]\le\frac{\mathbb{V}\left[\mathcal{R}\right]}{\left(2\left|j-j'\right|-1\right)^{2}\cdot\epsilon^{2}}~.
	\end{align*}
	By union bound, the probability that for some $j'\ne j$,
	$\mathcal{R}_{j'}\ge d(t_{j'},v_{j})$ is bounded by
	\[
	\sum_{j\ne j'}\Pr\left[\mathcal{R}_{j'}\ge d(t_{j'},v_{j})\right]<\frac{\mathbb{V}\left[\mathcal{R}\right]}{\epsilon^{2}}\cdot2\cdot\sum_{i=1}^{\infty}\frac{1}{i^{2}}=\frac{\mathbb{V}\left[\mathcal{R}\right]}{\epsilon^{2}}\cdot\frac{\pi^{2}}{3}~.
	\]
	We conclude
	\begin{align*}
	\Pr\left[X_{j}\right] & \ge\Pr\left[\mathcal{R}_{j'}\ge d(t_{j'},v_{j})\right]\cdot\left(1-\sum_{j\ne j'}\Pr\left[\mathcal{R}_{j'}\ge d(t_{j'},v_{j})\right]\right)\\
	& \ge\left(1-\frac{\mathbb{V}\left[\mathcal{R}\right]}{\epsilon^{2}}\right)\left(1-\frac{\mathbb{V}\left[\mathcal{R}\right]}{\epsilon^{2}}\cdot\frac{\pi^{2}}{3}\right)=1-O\left(\frac{1}{\epsilon^{2}\ln k}\right)=\Omega(1)~,
	\end{align*}
	for $\eps=\Theta(\frac{1}{\sqrt{\log k}})$. The Theorem now follows.
\end{proof}
Following the lines of the proof of \Cref{cor:VorLB}, we conclude:
\begin{corollary}
	Fix some $k\in\mathbb{N}$. There is a graph $G=(V,E,w)$ with terminal set $K$ of size $k$, such that with constant probability, the distortion of the minor returned by the \texttt{Ball-Growing} algorithm is $\Omega(\sqrt{\log k})$
\end{corollary}

\begin{remark}
	\Cref{thm:BGLB} can also be proved using concentration bounds. However, the lower bound remain $\Omega(\sqrt{\log k})$ so we provided the more basic proof using Chebyshev inequality.
	Nevertheless, the curious reader can find the required concentration bounds for such a proof in \Cref{appendix:ConcentrationBounds}.  
\end{remark}

\section{Discussion}\label{sec:Discussion}
In this paper we proved an $O(\log k)$ upper bound for the Steiner Point Removal problem, improving the previous $O(\log^2 k)$ upper bound by \cite{Che18}.
The lower bound is still only $8$ \cite{CXKR06}. Closing this gap remains an intriguing open problem.
Both the \texttt{Noisy-Voronoi} and the \texttt{Ball-growing} algorithms proceed by creating random terminal partitions. These partitions are determined using random parameters, which are chosen with no consideration whatsoever of the input graph $G$.
At contrast, the optimal tree algorithm of \cite{G01} is a deterministic recursive algorithm which make decisions after considering the tree structure at hand. 
It seems that the input-oblivious approach of the  \texttt{Noisy-Voronoi} and the \texttt{Ball-growing} algorithms is doomed for failure, and in fact, both these algorithms already fail to achieve constant distortion on a simple tree example.
As a conclusion, input-sensitive approaches seem to be more promising for future attempts to resolve the SPR problem.

We would like to emphesis two additional open problems:
\begin{itemize}
	\item Expected distortion: Currently the state of the art for usual (worst-case) distortion, and expected distortion for the SPR problem is the same. Both have $O(\log k)$ upper bound and $\Omega(1)$ lower bound. There are cases where much better results can be achieved for expected distortion (e.g. embed a graph into a tree must incur distortion  $\Omega(n)$, while a distribution over embeddings into trees can have expected distortion $O(\log n)$ \cite{FRT04}). What are the right bounds for expected distortion in the SPR problem?
	\item Special graph families: \cite{BG08} showed that constant distortion for the SPR problem can be achieved on outer-planar graphs. It will be very interesting to achieve better upper bounds for planar graphs, and more generally for minor-free graphs, bounded treewidth graphs etc. In the expected distortion regime, an $O(1)$ upper bound is already known \cite{EGKRTT14} for minor-free graphs.
\end{itemize}

\section{Acknowledgments}
The author would like to thank his advisors: to Ofer Neiman, for fruitful discussions, and to Robert Krauthgamer for useful comments.
	{\small
		\bibliographystyle{alpha}
		\bibliography{SteinerBib}
	}
	
	\appendix
	
\section{Concentration Bounds for Sum of Exponential Distributions}\label{appendix:ConcentrationBounds}
\begin{lemma}\label{lem:TightExpConcentration}
	Suppose $X_{1},\dots,X_{n}$'s are independent random
	variables, where each $X_{i}$ is distributed according to $\Exp(\lambda_{i})$.
	Let $X=\sum_{i}X_{i}$ and $\lambda_{M}=\max_i\lambda_{i}$. Set $\mu=\mathbb{E}\left[X\right]=\sum_{i}\lambda_{i}$. \\
	For $0<t\le\frac{1}{2\lambda_{M}}$, and $\alpha\ge 2t\lambda_M$:
	\begin{align*}
	\Pr\left[X\ge(1+\alpha)\mu\right] & \le\exp\left(-t\mu\cdot\left(\alpha-2t\lambda_M\right)\right)~.\\
	\Pr\left[X\le(1-\alpha)\mu\right] & \le\exp\left(-t\mu\left(\alpha-t\lambda_{M}\right)\right)~.
	\end{align*}
\end{lemma}
\begin{proof}
	For each $X_{i}$, the moment generating function w.r.t $t$ equals  
	$$\mathbb{E}\left[e^{tX_{i}}\right]=\frac{1}{1-t\lambda_{i}}=1+t\lambda_{i}\left(\sum_{\ell\ge0}\left(t\lambda_{i}\right)^{\ell}\right)\le1+t\lambda_{i}\left(1+2t\lambda_{i}\right)\le e^{t\lambda_{i}\left(1+2t\lambda_{i}\right)}~.$$
	Using Markov inequality,
	\begin{align*}
	\Pr\left[X\ge(1+\alpha)\mu\right] & =\Pr\left[e^{tX}\ge e^{t(1+\alpha)\mu}\right]\\
	& \le\mathbb{E}\left[e^{tX}\right]\cdot e^{-t(1+\alpha)\mu}\\
	& =e^{-t(1+\alpha)\sum_{\ell}\lambda_{\ell}}\cdot\prod_{\ell}\mathbb{E}\left[e^{tX_{\ell}}\right]\\
	& \le e^{-(1+\alpha)\sum_{\ell}t\lambda_{\ell}}\cdot e^{\sum_{\ell}t\lambda_{\ell}\left(1+2t\lambda_{\ell}\right)}\\
	& =e^{\sum_{\ell}\left(t\lambda_{\ell}\cdot\left(2t\lambda_{\ell}-\alpha\right)\right)}\\
	& \le e^{\left(\sum_{\ell}t\lambda_{\ell}\right)\cdot\left(2t\lambda_{M}-\alpha\right)}=e^{-t\mu\cdot\left(\alpha-2t\lambda_{M}\right)}~.
	\end{align*}
	where in the second
	equality we use the fact that $\left\{ X_{i}\right\} _{i}$ are independent.
	
	For the second inequality, it holds that:
	\[
	\mathbb{E}\left[e^{-tX_{i}}\right]=\frac{1}{1+t\lambda_{i}}=\sum_{\ell\ge0}\left(-1\right)^{\ell}\left(t\lambda_{i}\right)^{\ell}\le1-t\lambda_{i}\left(1-t\lambda_{i}\right)\le e^{-t\lambda_{i}\left(1-t\lambda_{i}\right)}~.
	\]
	Therefore,
	\begin{align*}
	\Pr\left[X\le(1-\alpha)\mu\right] & =\Pr\left[e^{-tX}\ge e^{-t(1-\alpha)\mu}\right]\\
	& \le\mathbb{E}\left[e^{-tX}\right]/e^{-t(1-\alpha)\mu}\\
	& =e^{t(1-\alpha)\mu}\cdot\Pi_{\ell}\mathbb{E}\left[e^{-tX_{\ell}}\right]\\
	& \le e^{(1-\alpha)\sum_{\ell}t\lambda_{\ell}}\cdot e^{-\sum_{\ell}t\lambda_{\ell}\left(1-t\lambda_{\ell}\right)}\\
	& =e^{-\sum_{\ell}t\lambda_{\ell}\left(\alpha-t\lambda_{\ell}\right)}\\
	& \le e^{-t\mu\left(\alpha-t\lambda_{M}\right)}~.
	\end{align*}
\end{proof}

We derive the following corollary. 
\begin{corollary}\label{cor:TightExpConcentration}
	Suppose $X_{1},\dots,X_{n}$ are independent random
	variables, where $X_{i}\sim\Exp(\lambda_{i})$.
	Let $X=\sum_{i}X_{i}$ and $\lambda_{M}=\max_i\lambda_{i}$. Set $\mu=\mathbb{E}\left[X\right]=\sum_{i}\lambda_{i}$. Then:
	\begin{align*}
	\text{For }\alpha\le2:\,\,\, & \Pr\left[X\ge(1+\alpha)\mu\right]\le\exp\left(-\frac{\alpha^{2}\mu}{8\lambda_{M}}\right)~.\\
	\text{For }\alpha\le1:\,\,\, & \Pr\left[X\le(1-\alpha)\mu\right]\le\exp\left(-\frac{\alpha^{2}\mu}{4\lambda_{M}}\right)~.
	\end{align*}
\end{corollary}
For the first inequality we choose the parameter $t=\frac{\alpha}{2}\cdot\frac{1}{2\lambda_M}$, while for the second inequality we choose the parameter $t=\alpha\cdot\frac{1}{2\lambda_M}$.

\section{The \texttt{Ball-Growing} Algorithm}\label{sec:BallGrowing}
The \texttt{Ball-Growing} algorithm assumes w.l.o.g that the minimal distance between terminal to a Steiner vertex in the input graph is exactly $1$.
Throughout the execution of the algorithm each terminal $t_j$, is associated with a radius $R_j$ and cluster $V_j\subset V$. 
The algorithm iteratively grow clusters $V_1,\dots,V_k$ around the terminals. Once some vertex $v$ joins some cluster $V_j$, it will stay there. When all the vertices are clustered, the algorithm terminates. 
Initially the cluster $V_j$ contains only the terminal $t_j$, while $R_j$ equals $0$.
The algorithm will have rounds, where each round consist of $k$ steps. In step $j$ of round $\ell$, the algorithm samples a number $q_j^\ell$ according to distribution $\Exp(D\cdot r^\ell)$ (note that the mean of the distribution grows by a factor of $r$ in each round). The radius $R_j$ grows by  $q_j^\ell$. We consider the graph induced by the unclustered vertices $V_\perp$ union $V_j$. Every unclustered vertex of distance at most $R_j$ from $t_j$ in $G[V_\perp\cup V_j]$ joins $V_j$.

\begin{algorithm}[!ht]
	\caption{$M=\texttt{Ball-Growing}(G=(V,E),w,K=\{t_1,\dots,t_k\})$}\label{alg:BG}
	\begin{algorithmic}[1]
		\STATE Set $r\la1+\delta / \ln k$, where $\delta = \nicefrac{1}{20}$.
		\STATE Set $D\la \frac \delta{\ln k}$.
		\STATE For each $j\in [k]$, set $V_j\la\{t_j\}$, and set $R_j\la~ 0$.
		\STATE Set $V_\perp~\la~V\setminus\left(\cup_{j=1}^k V_j\right)$.
		\STATE Set $\ell ~\la~ 0$.
		
		\WHILE{$\left(\cup_{j=1}^k V_j\right) ~\neq~ V$}
		\FOR {$j$ from $1$ to $k$}
		\STATE Choose independently at random $q^\ell_j$ distributed according to $\Exp(D\cdot r^\ell)$.
		\STATE Set $R_j\la R_j+q^\ell_j$.
		\STATE Set $V_j\la B_{G[V_\perp \cup V_j]}(t_j,R_j)$.	.\hfill\emph{//~This is the same as $V_j\la V_j\cup B_{G[V_\perp \cup V_j]}(t_j,R_j)$.}					
		\STATE Set $V_\perp\la V\setminus\left(\cup_{j=1}^k V_j\right)$.
		\ENDFOR
		\STATE $\ell\la\ell+1$.
		\ENDWHILE
		\RETURN the terminal-centered minor $M$ of $G$ induced by $V_1,\ldots,V_k$.
	\end{algorithmic}	
\end{algorithm}

\newpage
\begin{multicols}{2}
\footnotesize 
\section{Index}\label{appendix:key}
\vspace{-3pt}
\subsubsection*{Preliminaries}
\begin{description}
	\item[$d_{G}$] : shortest path metric in $G$.
	\item[{$G[A]$}] : graph induced by $A$.
	\item[$K$] $=\{t_1,\dots,t_k\}$ : set of terminals.
	\item[$D(v)$] $=\min_{t\in K}d_{G}(v,t)$.
	\item[Terminal partition] : partition  $\{V_1,\dots,V_k\}$ of $V$, s.t. for every i, $t_i\in V_i$ and $V_i$ is connected.
	\item[Induced minor] : given terminal partition $\{V_1,\dots,V_k\}$, the induced minor obtained by contracting each $V_i$ into the super vertex $t_i$. The weight of the edge $\{t_i,t_j\}$ (if exist) set to be $d_G(t_i,t_j)$.
	\item[Distortion] of induced minor: $\max_{i,j}\frac{d_M(t_i,t_j)}{d_G(t_i,t_j)}$.
	\item[$\Geo(p)$] : geometric distribution with parameter $\lambda$.	
	\item[$\Exp(\lambda)$] : exponential distribution with parameter $p$.
\end{description}
\vspace{-3pt}
\subsubsection*{Modification}
	Every edge on $P_{t,t'}$ has weight at most $c_{w}\cdot d_{G}(t,t')$.
\vspace{-3pt}
\subsubsection*{Constants}
\begin{description}
	\item[$p$] $=\frac15$: parameter of the geometric distribution.
	\item[$\delta$] $=\frac{1}{20\cdot\ln k}$: jumps in $R_j$ are of magnitude $1+\delta$
	\item[$c_{w}$] $=\frac{\delta}{24}.$
	\item[$\cint$] $=\frac{1}{6}$:
	governs the size of interval in the partition $\mathcal{Q}$ of $P_{t,t'}$.
	\item[$c_{\text{con}}$] $=\frac12$: used to bound the variation of the charge function from its expectation.
	\item[$c_d$] $=e^2$: bound on the maximal size of $R_j$.
\end{description}
\vspace{-3pt}
\subsubsection*{Events}
\begin{description}
	\item[$\EfBig$] : denotes that for some pair of terminals $t,t'$, $f(\{X(Q)\}_{Q\in\mathcal{Q}}\ge43\cdot d_{G}(t,t')$.
	\item[$\EB$] : denotes that the exist $j$, such that $R_j>c_d$.
\end{description}

\subsubsection*{Notations}
\begin{description}
	\item[$V_{j}$] : cluster of $t_{j}$.
	\item[$R_{j}$] : magnitude of the cluster of $t_{j}$.
	\item[$V_{\perp}$] : set of unclustered (uncovered) vertices.
	
	\item[$P_{t,t'}$] $=\left\{ t=v_{0},\dots,v_{\gamma}=t'\right\} $: shortest
	path from $t$ to $t'$.
	\item[$L(\left\{ v_{a},v_{a+1},\dots,v_{b}\right\} )$] $=d_{G}(v_{a},v_{b})$:
	internal length.
	\item[$ L^{+}(\left\{ v_{a},v_{a+1},\dots,v_{b}\right\} )$] $=d_{G}(v_{a-1},v_{b+1})$:
	external length.
	\item[$\mathcal{Q}$] : partition of $P_{t,t'}$ into intervals $Q$.
	\item[$a_{j}$] : the leftmost active vertex covered by $t_{j}$.
	\item[$b_{j}$] : the rightmost active vertex covered by $t_{j}$.
	\item[$\mathcal{D}_{j}$] $=\left\{ a_{j},\dots,b_{j}\right\} $:
	detour created by terminal $t_{j}$.
	\item[Slice]  maximal sub-interval (of some $Q$) of active vertices.
	\item[$r_{v}$] : minimal choice of $R_{j}$, such that $v$ joins $V_j$.
	\item[$v^{j}$] : vertex with the minimal $r_{v}$ (among active vertices).
	\item[$Q_{j}$] : interval containing $v_{j}$.
	\item[$S_{j}$] : slice containing $v_{j}$.
	\item[$f(\{x_Q\}_{Q\in\mathcal{Q}})$] : $=\sum_{Q\in\mathcal{Q}}x_Q\cdot L^{+}(Q)$,
	charge function.
	\item[$B_{Q}$] : a coin box which resembles the interval $Q$.
	\item[$d_{G,V_i+V_j}(t_i,t_j)$] : The weight of the shortest path in $G$ between $t_1$ and $t_2$ that uses only vertices from $V_i\cup V_j$, and only a single crossing edge between $V_i$ to $V_j$. 
\end{description}
\vspace{-3pt}
\subsubsection*{Counters}
\begin{description}
	\item[$\mathcal{S}(Q)$] : (current) number of slices in interval $Q$.
	\item[$X(Q)$] : number of detours the interval $Q$ is (currently) charged for.
	\item[$\tilde{X}(Q)$] : number of detours the interval $Q$ is charged for
	by the end of \Cref{alg:mainSPR}.
	\item[$Z(Q)$] : number of active coins in $B_{Q}$. Each coin is active	when added to the box.
	\item[$Y(Q)$] : number of inactive coins in $B_{Q}$. A coin become inactive after tossing.
	\item[$\tilde{Y}(Q)$] : number of inactive coins in $B_{Q}$ by the end of the process.
\end{description}

\end{multicols}

\end{document}